\documentclass[11pt]{article}

\usepackage[english]{babel}
\usepackage{amsfonts}
\usepackage[inline]{enumitem}
\usepackage{tikzducks}
\usetikzlibrary{positioning}

\usepackage{geometry}
\usepackage{amsmath}
\usepackage{float} 

\usepackage{graphicx}
\usepackage[colorlinks=true, allcolors=blue]{hyperref}
\usepackage{csquotes}

\usepackage{amsthm}

\newtheorem{theorem}{Theorem}[section]
\newtheorem{corollary}[theorem]{Corollary}
\newtheorem{lemma}[theorem]{Lemma}
\newtheorem{proposition}[theorem]{Proposition}
\newtheorem{claim}[theorem]{Claim}
\newtheorem{definition}[theorem]{Definition}

\usepackage{graphicx}
\usepackage{subcaption}
\usepackage{tikz}

\title{Impossibilities for Obviously Strategy-Proof Mechanisms}
\author{Shiri Ron\thanks{Weizmann Institute of Science. Email:  shiriron@weizmann.ac.il. Supported by an Azrieli Foundation fellowship, ISF grant 2185/19, and BSF-NSF grant (BSF number: 2021655, NSF number: 2127781).
Part of the work was done while the author was at residence in Simons Laufer Mathematical Sciences Institute, supported by the National Science Foundation (grant number DMS-1928930). 
}}
\begin{document}
\maketitle

\begin{abstract}
We explore the approximation power of deterministic obviously strategy-proof mechanisms in auctions, where the objective is welfare maximization. A trivial ascending auction on the grand bundle guarantees an approximation of $\min\{m,n\}$ for all valuation classes, where $m$ is the number of items and $n$ is the number of bidders. We focus on two classes of valuations considered “simple”: additive valuations and unit-demand valuations. For additive valuations, Bade and Gonczarowski [EC'17] have shown that exact welfare maximization is impossible. No impossibilities are known for unit-demand valuations.

We show that if bidders' valuations are additive or unit-demand, then no obviously strategy-proof mechanism gives an approximation better than $\min\{m,n\}$. Thus, the aforementioned trivial ascending auction on the grand bundle is the optimal obviously strategy-proof mechanism. These results illustrate a stark separation between the power of dominant-strategy and obviously strategy-proof mechanisms. The reason for it is that for both of these classes the dominant-strategy VCG mechanism does not only optimize the welfare exactly, but is also “easy” both from a computation and communication perspective.

In addition, we prove tight impossibilities for unknown single-minded bidders in a multi-unit auction and in a combinatorial auction. We show that in these environments as well, a trivial ascending auction on the grand bundle is optimal.
\end{abstract}
\thispagestyle{empty}
\clearpage
\setcounter{page}{1}

\section{Introduction}
Consider a second-price auction of a single good:
there is one item, and each bidder $i$ has a value of $x_i$ for the item. Our goal is to give the item to the bidder with the highest value.  We assume that bidders are strategic, meaning that they aim to maximize their {utility}, which is the value they derive from the item minus their payment. 
In a second-price auction, the highest bidder wins the item and pays the value reported by the second highest bidder \cite{Vic61}. By allocating the item to the highest bidder, we maximize the social welfare. 
The implementation where all bidders send their bids simultaneously,
known as the sealed-bid auction,
is dominant-strategy incentive compatible.  

Dominant strategy mechanisms are desirable for multiple reasons: 
they are considered straightforward for participants to understand and follow. 
Thus, the cognitive burden of participating in them is reduced. Another advantage of dominant-strategy mechanisms is their predictability. If
each participant chooses their dominant strategy, the mechanism's behavior becomes easier to anticipate, which is beneficial for decision-making and planning purposes.
Despite being a dominant-strategy mechanism, the sealed-bid second price auction is rarely used in practice 
 (see, e.g., \cite{MM87,ausubel04,ausobel2005}). 
 
However,  the sealed-bid implementation is not the only way to realize a second-price auction. 
Another realization, which is far more prevalent in real-life scenarios \cite{ausubel04,li} is an ascending auction: 
the price gradually increases and players report at each round $r$ whether they are interested in the item given the current price $p_r$. The ascending auction ends when there is one bidder left. This bidder wins the item and pays the price presented in the last round. This implementation is also dominant strategy incentive compatible, similarly to the sealed-bid implementation.  

Even though both auctions satisfy the guarantee of dominant strategy incentive compatibility,
it is generally believed
 that players strategize less in an ascending auction compared to a sealed-bid auction 
(see, e.g., \cite{li,kagel-harstad-levin,ausubel04}). 
\cite{li} explains it as follows: although the sealed-bid implementation of a second-price auction satisfies that reporting the true value of the item is a dominant strategy for all players, this fact is neither obvious nor easy to explain.   
If the participants fail to comprehend that they should follow their dominant strategies, 
the mechanism  loses its desirable qualities. In contrast, the ascending auction is self-explanatory: bidders
readily see that strategizing in it is not beneficial for them.

To shed light on this phenomenon, \cite{li}
has introduced the concept of \emph{obviously strategy-proof mechanisms}, which distinguishes sealed-bid and ascending auctions.  Roughly speaking, obviously strategy-proof mechanisms have dominant strategies that satisfy a special property: it is evident, even for cognitively limited agents, that these strategies are dominant.  
What makes the effectiveness of the dominant strategy in these mechanisms so readily apparent 
is that the worst-case scenario when following this strategy is
always at least as beneficial as the best-case scenario when using any other strategy. 

Due to the appealing nature of 
of obviously strategy-proof mechanisms, they have been studied
 in various contexts, e.g.
 matching  \cite{AG18,T19,PS22,T21},
scheduling \cite{KV19,FMPV19,FV21,FMPV23}
and allocation problems \cite{BG17,PT19,KV19,FMPV19,dKKV20,FV21,FPV21,PS22,FMPV23,FV23}. 
The goal of this paper is to understand the approximation power of obviously strategy-proof mechanisms in auctions.

The case of an auction with a single item is fully understood: 
an ascending implementation of the second-price auction achieves the optimal welfare and is obviously strategy-proof. 
But what happens if there is more than one item in the auction? Is there an obviously strategy-proof mechanism that 
maximizes or  approximates the optimal social welfare?
In other words, what is the power of obviously strategy-proof mechanisms in combinatorial auctions? This is the main question that we explore in this work.  

A combinatorial auction consists of a set of players that we denote with $N$, where $|N|=n$, and a set of  items that we denote with $M$ $(|M|=m)$. 
Each player $i$ has a valuation function $v_i:2^M\to \mathbb{R}^+$ that specifies her value for every subset of items. 
We assume that all valuation functions are \emph{normalized} $(v_i(\emptyset)=0)$ and \emph{monotone} (meaning that for every two bundles of items $T_1,T_2\subseteq M$, if $T_2$ contains  $T_1$, then $v_i(T_2)\ge v_i(T_1)$). The valuation function of each bidder is private and belongs to a domain of valuations $V_i$. We sometimes impose additional restrictions on $V_i$. Our goal is to allocate the items to the bidders in a way that maximizes the social welfare, which is the sum of the values of the bidders given their bundles. We assume that the bidders aim to maximize their \emph{utility}, which is the value from the bundle that they get minus their payment.  

Existing literature about obviously strategy-proof mechanisms in combinatorial auctions focused on the case of single-minded bidders. A valuation is
\emph{single-minded} if there is a single bundle $S^\ast\subseteq M$ and a scalar $\alpha>0$ such that:
$$
\forall X\subseteq M,\quad v(X)=\begin{cases}
\alpha, \quad S^\ast \subseteq X, \\
0, \quad \text{otherwise.}
\end{cases}
$$
A bidder is \emph{single-minded} if all the valuations in her domains are single-minded. 
If all the valuations in the domain of bidder $i$ are parameterized with the same bundle, then we say that she is \emph{known single-minded}. If this is not the case, she is \emph{unknown single-minded}. 

\cite{dKKV20} present a deterministic obviously strategy-proof mechanism that gives an approximation of $\mathcal O(\sqrt{m})$ for known single-minded bidders, and another mechanism with an approximation of $k$ for unknown single-minded bidders, if the largest set desired is of size $k$.

For the simple classes of additive bidders and unit-demand bidders, not much is known. 
\cite{BG17} have shown that even for   two items and two additive bidders,
 there is 
no obviously strategy-proof mechanism that maximizes the social welfare.\footnote{This impossibility holds only for mechanisms that satisfy the standard assumption that bidders that gain no items pay nothing.} 
No impossibilities are known for unit-demand bidders. For both of these classes, achieving a $\min\{m,n\}$ approximation is trivial, using  an
\emph{ascending auction on the grand bundle}: taking all the items in $M$ and running an ascending auction on all of them together.

\subsection*{Our Results I: Impossibilities for Combinatorial Auctions}
We show that for several classes of valuations, including simple valuation classes such as additive and unit demand valuations, no deterministic obviously strategy-proof mechanism has an approximation guarantee that is strictly better than $\min\{m,n\}$. Since this is the approximation ratio obtained by the ascending auction on the grand bundle, all of our impossibility results are tight. 
In addition, all the impossibilities hold for mechanisms with  unbounded computation and communication. 

Our impossibilities hold for mechanisms that satisfy individual rationality and no-negative-transfers, meaning that the bidders' utility is non-negative and that bidders do not get paid by the mechanism. 
We begin our explorations by considering unit-demand bidders:
\begin{theorem}\label{thm-lb-unit}
An auction with $m\ge 2$ items and $n\ge 2$ unit-demand bidders 
has  no
 obviously strategy-proof mechanism 
 that satisfies
 individual rationality and no-negative-transfers
 and
 gives an approximation strictly  better than  $\min\{m,n\}$ to the social welfare.
\end{theorem}

This is the first impossibility for obviously strategy-proof mechanisms for unit-demand bidders. An immediate corollary of Theorem \ref{thm-lb-unit} is that the impossibility extends to every class of valuations that contains unit-demand valuations. In particular, the impossibility holds for all the classes presented in the hierarchy of complement-free valuations of \cite{LLN01-journal} and for arbitrary monotone valuations. 
\begin{corollary}
    An auction with $m\ge 2$ items and $n\ge 2$ bidders with  \emph{gross substitute valuations} or \emph{submodular valuations} or \emph{fractionally subadditive valuations} or \emph{subadditive valuations} or \emph{arbitrary monotone valuations}
    has no
     obviously strategy-proof mechanism 
 that satisfies
 individual rationality and no-negative-transfers
and gives an approximation strictly  better than  $\min\{m,n\}$ to the social welfare.\end{corollary} 

Our second main result concerns the case of additive bidders. For this class, \cite{BG17} prove that exact welfare maximization is impossible (under the assumption that losers pay zero). We show that:  
\begin{theorem}\label{thm-lb-add}
An auction with $m\ge 2$ items and $n\ge 2$ additive bidders 
has  no
 obviously strategy-proof mechanism 
 that satisfies
 individual rationality and no-negative-transfers
 and
 gives an approximation strictly  better than  $\min\{m,n\}$ to the social welfare. 
\end{theorem} 

We proceed by showing that we can circumvent the impossibility in Theorem \ref{thm-lb-add} for a restricted class of additive valuations: 
\begin{theorem}\label{thm-ub}
Let $x_h>x_l$ be two positive scalars.  Then, for every auction with additive bidders whose values for each item belong in the set $\{0,x_{l},x_{h}\}$, there exists an obviously strategy-proof mechanism that maximizes the social welfare.     
\end{theorem}
The obviously strategy-proof mechanism of Theorem \ref{thm-ub} 
is a posted-price mechanism, where  
the price of every item is set to be $x_l$. The mechanism has two rounds: in the first round, each player takes only items for which she has a value of $x_h$, and in the second one she takes the remaining items for which she has a value of $x_l$. 
Informally, what makes this mechanism 
obviously strategy-proof is the fact that bidders increase their profit only by taking items for which their value is $x_h$. 

We now switch gears by proving impossibilities for unknown single-minded bidders:  
\begin{theorem}\label{cor-ca-sm}
A combinatorial auction  with $m\ge 2$  items and $n\ge 2$ unknown single-minded bidders    has no
     obviously strategy-proof mechanism 
 that satisfies
 individual rationality and no-negative-transfers
and gives an approximation strictly  better than  $\min\{m,n\}$ to the social welfare.
\end{theorem}
We remind that \cite{dKKV20} 
provide an obviously strategy-proof mechanism that $k$-approximates the welfare with unknown single-minded bidders, where $k$ is the size of the largest desirable set.  Indeed, in the proof of Theorem \ref{cor-ca-sm}, we use valuations whose demanded set is of size $m$, so our proof shows that the dependence of the approximation ratio  on the size of the largest desired set is necessary.  

Another implication of Theorem \ref{cor-ca-sm} is that obviously strategy-proof mechanisms are more powerful
for known single-minded bidders compared to unknown single-minded bidders: for known single-minded bidders, there is an 
$\mathcal{O}(\sqrt m)$ approximation to the welfare \cite{dKKV20}, whilst we show that for unknown single-minded bidders, obviously strategy-proof mechanisms cannot get an approximation better than $\min\{m,n\}$.  
\subsection*{Our Results II: Impossibilities for Multi-Unit Auctions}
We wrap up by providing an
equivalent impossibility for unknown single-minded bidders in the multi-unit auction.  
In the multi-unit auction setting, often called the knapsack auction, items are identical 
and the valuation function of each player $v_i:[m]\to \mathbb R^+$ maps quantities of items to values. 

For multi-unit auctions, i.e., auctions where all the items are identical, two special cases were previously considered. The first one is the \textquote{single-parameter} case, with known single-minded bidders. 
For this class of valuations, there exist obviously strategy-proof mechanisms that give ${\mathcal O}(\min\{\log m,\log n\})$ approximation to the welfare \cite{DGR14,CGS23}, and no mechanism  gives an approximation better than $\Omega(\sqrt{\log n})$ of the welfare \cite{FPV21}.  
For valuations that exhibit deceasing marginal values,\footnote{A valuation $v$ satisfies \emph{decreasing marginal values} if for every quantity $j\in \{0,\ldots,m-1\}$, $v(j+1)-v(j)\le v(j)-v(j-1).$} \cite{GMR17} have shown an obviously strategy-proof clock auction that gives an approximation of $\mathcal O(\log n)$ to the welfare and have claimed that no obviously strategy-proof mechanism gives an approximation better than $\sqrt{2}$ to the welfare.

We now prove an impossibility for unknown-single-minded bidders:
\begin{theorem}\label{thm-lb-mua}
 A multi-unit auction with $m\ge 2$  items and $n\ge 2$ unknown single-minded bidders 
   has no
     obviously strategy-proof mechanism 
 that satisfies
 individual rationality and no-negative-transfers
and gives an approximation strictly  better than  $\min\{m,n\}$ to the social welfare.  
\end{theorem}
Similarly to combinatorial auctions, Theorem \ref{thm-lb-mua} implies that obviously strategy-proof mechanisms are more powerful for known single-minded than for unknown single-minded (because of the obviously strategy-proof mechanism of \cite{DGR14} for known single-minded bidders that gives $\Theta(\log m)$ approximation to the welfare).
Moreover, Theorem \ref{thm-lb-mua} implies an impossibility for the class of arbitrary monotone valuations:
\begin{corollary}\label{cor-mua-sm}
A multi-unit auction with $m\ge 2$  items and $n\ge 2$ bidders with arbitrary monotone valuations    has no
     obviously strategy-proof mechanism 
 that satisfies
 individual rationality and no-negative-transfers
and gives an approximation strictly  better than  $\min\{m,n\}$ to the social welfare.
\end{corollary}
The impossibilities in Theorem \ref{thm-lb-mua} and Corollary \ref{cor-mua-sm} are tight since the ascending auction on the grand bundle gives $\min\{m,n\}$ approximation to the welfare. 
\subsection*{Open Questions} First, we have yet to understand the power of obviously strategy-proof mechanisms in \textquote{single-parameter} domains.\footnote{For a discussion of the definition of single-parameter domains, see \cite{BDR23}.} 
In particular, we do not know yet whether it is possible to improve  the approximation guarantee of $O(\sqrt m)$ of the mechanism of \cite{dKKV20} for known single-minded bidders. 
In addition, for multi-unit auctions with decreasing marginal values, it is still uncertain whether obviously strategy-proof mechanisms can give a constant approximation of the welfare or not.

\subsection*{Structure of the Paper}
In Section \ref{sec-prem}, we give some necessary preliminaries. Subsequently, we 
prove impossibilities for unknown single-minded bidders both in a combinatorial auction and in a multi-unit auction
(Section \ref{sec-mua}). In Section \ref{sec-additive-unit}, we explore obviously strategy-proof mechanisms for additive valuations: we provide an impossibility and point out that for a sufficiently small set of valuations, there exists an obviously strategy-proof mechanism. The proof of the impossibility for unit-demand valuations, which is very similar to the proof for additive valuations, is  in Appendix \ref{appsec-unit}.

\section{Preliminaries and Basic Observations}\label{sec-prem}
We use communication protocols to represent mechanisms. A protocol is  visualized as a tree, specifying which player speaks  at each node and determining the subsequent node based on the message. 
 Each leaf in the protocol is associated with an allocation of the items to the bidders and with a payment for each player.  Throughout this paper, we discuss only deterministic protocols. 
Throughout the paper, we assume that all mechanisms satisfy \emph{perfect information}, i.e., that all messages sent are observed by all the agents,\footnote{In the communication complexity literature, this assumption is called the blackboard model.} and also that they are  \emph{sequential}, meaning that exactly one bidder sends a message at each node. We denote with $\mathcal N_i$ all the nodes in which player $i$ sends a message.\footnote{All the results in this paper hold regardless of these assumptions. 
We can assume perfect information without loss of generality because of the revelation principle of  
\cite{BG17} for obviously strategy-proof mechanisms. As for the assumption of sequential mechanisms, it is easy to verify that every social choice function implemented by an obviously strategy-proof mechanism that allows players to speak simultaneously can also be implemented by a sequential obviously strategy-proof mechanism.}
For ease of presentation, we define obviously strategy-proof and dominant strategy mechanisms only for perfect information and sequential mechanisms.  

We now delve deeper into the intricate details of mechanisms by  defining strategies and behaviors. 
Given a mechanism, a \emph{behavior} $B_i$ of player $i$ specifies a message for each node in $\mathcal N_i$. Denote with $\mathcal B_i$ the set of all possible behaviors of player $i$. Note that each behavior profile $B=(B_1,\ldots,B_n)$ corresponds to a path of all the visited nodes given this behavior profile. We denote with $Path(B)$ all the nodes in this path and with $Leaf(B)$ the leaf at the end of this path.  In addition, let us denote by $f_i(B)$ and $p_i(B)$ the bundle and the payment of player $i$ that are associated with $Leaf(B)$.  See Figure \ref{fig-prems} for an illustration.

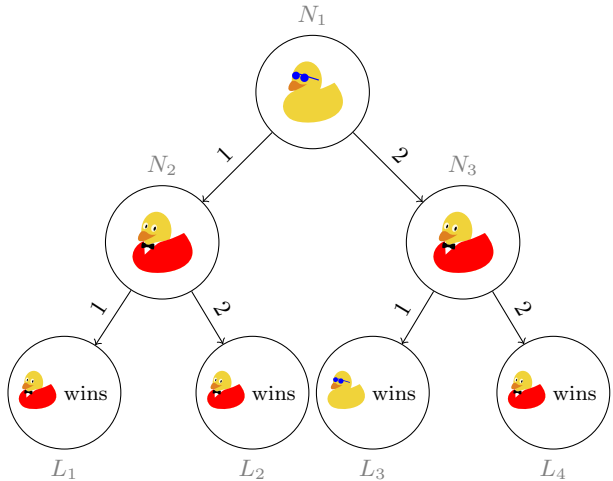
\begin{figure}[ht]
    \centering
    \captionsetup{parindent=\parindent}
    \caption[LoF entry]{\footnotesize 
    An illustration of a tree protocol that realizes a second-price auction of a single item with two bidders: a duck with sunglasses and a duck in a jacket. The values for the item of both of them belong in the set $\{1,2\}$, and we assume that the jacket duck wins the item in case of a tie. The payments are omitted for simplicity. 

 The behavior set  $\mathcal B_{s}$ of the sunglasses-wearing duck  specifies the possible messages she can send at node $N_1$, i.e.,  $\mathcal B_{s}=\{(N_1:"1"),(N_1:"2")\}$.  Analogously,   the set of behaviors $\mathcal B_{j}$ of the jacket-wearing duck
describes the message combinations at nodes $N_2$ and $N_3$, so  $\mathcal B_j=\{
(N_2:"1",N_3: "1"),(N_2:"1",N_3:"2"),(N_2:"2",N_3:"1"),(N_2:"2",N_3:"2")\}$. 
    For illustration, consider the behavior profile  $B=\{B_s=(N_1:"2"),B_j=(N_2:"2",N_3:"1")\}$. Note that $Path(B)=\{N_1,N_3,L_3\}$ and $Leaf(B)=L_3$.  
    }
     \label{fig-prems}
    				\begin{tikzpicture}
			\node[shape=circle,draw=black,minimum size=1.1cm,
   ,label=
   {[text=gray]above:\footnotesize $N_1$}]
    (u) at (2,2) {
        \begin{tikzpicture}[scale=0.40]
            \duck[sunglasses=blue];
        \end{tikzpicture}};
			\node[shape=circle,draw=black,minimum size=1.1cm, label=
   {[text=gray]above:\footnotesize $N_2$}](v) at (0,0){  \begin{tikzpicture}[scale=0.4]
            \duck[tshirt,jacket=red,bowtie]
        \end{tikzpicture}};

 \node [shape=circle, draw=black, minimum size=1.5cm,label=
   {[text=gray]below:\footnotesize $L_1$}] (l1) at (-1.3,-2) {};

\duck[tshirt,jacket=red,bowtie,scale=0.25,xshift=-220,yshift=-255] (duck1) {};

 \node [right=0.1cm of l1,xshift=-1cm] (lwins) {\scriptsize wins};


   \node[shape=circle,draw=black,minimum size=1.1cm,label=
   {[text=gray]above:\footnotesize $N_3$}]
   (v') at (4,0){  \begin{tikzpicture}[scale=0.4]
            \duck[tshirt,jacket=red,bowtie]
        \end{tikzpicture}};
			\node[shape=circle,draw=black,minimum size=1.5cm,label=
   {[text=gray]below:\footnotesize $L_4$}] (l4) at (5.2,-2) {};
   \duck[tshirt,jacket=red,bowtie,scale=0.25,xshift=520,yshift=-255] (duck4) {};

    \node [right=0.1cm of l4,xshift=-1cm] (lwins4) {\scriptsize wins};
			\node[shape=circle,draw=black,minimum size=1.5cm,label=
   {[text=gray]below:\footnotesize $L_3$}] (l3) at (2.8,-2) {}; 

   \duck[sunglasses=blue,scale=0.25,xshift=247,yshift=-255] (duck3) {};

    \node [right=0.1cm of l3,xshift=-1cm] (lwins3) {\scriptsize wins};
			
			\node[shape=circle,draw=black,minimum size=1.5cm,label=
   {[text=gray]below:\footnotesize $L_2$}] (l2) at (1.2,-2) {}; 

   \duck[tshirt,jacket=red,bowtie,scale=0.25,xshift=65,yshift=-255] (duck2) {};

 \node [right=0.1cm of l2,xshift=-1cm] (lwins2) {\scriptsize wins};

			\draw [->] (u) edge  node[sloped, above] {\footnotesize $1$} (v);
			\draw [->] (u) edge  node[sloped, above] {\footnotesize $2$} (v');
			\draw [->] (v) edge  node[sloped, above] {\footnotesize $1$}  (l1);
			\draw [->] (v') edge[]  node[sloped, above] {\footnotesize $2$} (l4);
			\draw [->] (v) edge  node[sloped, above] {\footnotesize $2$} (l2);
			\draw [->] (v') edge  node[sloped, above] {\footnotesize $1$} (l3);

		\end{tikzpicture}
\end{figure}

Given a mechanism, a \emph{strategy} of player $i$ is a function $\mathcal S_i:V_i\to \mathcal{B}_i$  that takes as input a valuation of player $i$ and outputs a message for each node in $\mathcal N_i$. For example, the truthful strategy of the jacket-wearing duck in Figure \ref{fig-prems} maps the valuation that has a value $1$ for the item  to the behavior $(N_2:"1",N_3:"1")$ and the valuation that has a value of $2$ for the item to the behavior $(N_2:"2",N_3:"2")$.

We denote with $\mathcal T$ the set of all  possible allocations of the items  to the bidders.
A mechanism $\mathcal M$ together with strategies $(\mathcal S_1,\ldots,\mathcal S_n)$ realize a social choice function $f:V_1\times\cdots\times V_n\to \mathcal T$ together with payments $P_1,\ldots,P_n:V_1\times\cdots\times V_n\to \mathbb R^{n}$ if for every $(v_1,\ldots,v_n)\in V_1\times \cdots\times V_n$, $Leaf(\mathcal{S}_1(v_1),\ldots,\mathcal{S}_n(v_n))$ is labeled with the allocation that $f(v_1,\ldots,v_n)$ outputs and with the payment that $P_i(v_1,\ldots,v_n)$ specifies for every player $i$. 

Given a valuation profile $(v_1,\ldots,v_n)$, the \emph{social welfare} of an allocation $T=(T_1,\ldots,T_n)$ is $\sum_{i=1}^{n}v_i(T_i)$. We often abuse notation by writing  $v_i(T)$ instead of $v_i(T_i)$.  
A mechanism $\mathcal M$ together with strategies $(\mathcal S_1,\ldots,\mathcal S_n)$  $\alpha$-approximate the welfare if they realize a social choice function $f:V_1\times\cdots\times V_n\to\mathcal T$ such that for every valuation profile $(v_1,\ldots,v_n)\in V_1\times\cdots\times V_n$, the social choice function $f(v_1,\ldots,v_n)$ outputs an allocation $(T_1,\ldots,T_n)$ with social welfare which is at least $\frac{1}{\alpha}$ fraction of the welfare of the optimal allocation.

In the context of auctions, the following properties of mechanisms 
are desirable: individual rationality and no-negative-transfers. 
A mechanism $\mathcal M$ together with strategies $(\mathcal{S}_1,\ldots,\mathcal{S}_n)$ are \emph{individually rational} if they realizes a allocation rule $f$ together with payment schemes $P_1,\ldots,P_n$ such that  for every player $i$ and for every valuation profile $(v_1,\ldots,v_n)\in V_1\times\cdots\times V_n$:
$$
v_i(f(v_i,v_{-i})) - P_i(v_i,v_{-i}) \ge 0
$$
In other words, a player never \textquote{loses} by participating in the auction.
A mechanism $\mathcal M$ together with the strategy profile $(\mathcal{S}_1,\ldots,\mathcal{S}_n)$ satisfy  \emph{no-negative-transfers} if they realize payment schemes such that 
for every player $i$ and for every valuation profile $(v_1,\ldots,v_n)\in V_1\times\cdots\times V_n$, $P_i(v_i,v_{-i}) \ge 0$, 
indicating that no player is ever given money.

\subsubsection*{Dominant-Strategy Mechanisms} 
We now define dominant-strategy mechanisms. 
A strategy $\mathcal S_i^\ast$ is \emph{dominant} for player $i$ if for every strategy profile of the other players $\mathcal S_{-i}$, for every valuation profile $(v_1,\ldots,v_n)\in V_1\times\cdots\times V_n$ and for every alternative strategy $\widehat{\mathcal S_i}$: 
$$
v_i(f_i(\mathcal{S}_i^\ast(v_i),\mathcal S_{-i}(v_{-i}))) - p_i(\mathcal{S}_i^\ast(v_i),\mathcal S_{-i}(v_{-i})) \geq v_i(f_i(\widehat{\mathcal{S}_i}(v_i),\mathcal S_{-i}(v_{-i}))) - p_i(\widehat{\mathcal{S}_i}(v_i),\mathcal S_{-i}(v_{-i}))
$$
where we remind that  $f_i(\mathcal{S}_i^\ast(v_i),\mathcal{S}_{-i}(v_{-i}))$  and  $p_i(\mathcal{S}_i^\ast(v_i),\mathcal{S}_{-i}(v_{-i}))$ specify the allocation and payment of player $i$, respectively, when player $i$ follows the  actions specified by  $\mathcal{S}^\ast_i(v_i)$ and the other players follow  the actions specified in the vector  $\mathcal{S}_{-i}(v_{-i})=(\mathcal{S}_1(v_1),\ldots,\allowbreak \mathcal{S}_{i-1}(v_{i-1}), \mathcal S_{i+1}(v_{i+1}), \ldots,\allowbreak \mathcal{S}_n(v_n))$. The same holds for $f_i(\widehat{\mathcal{S}_i}(v_i),\mathcal S_{-i}(v_{-i}))$ and $p_i(\widehat{\mathcal{S}_i}(v_i),\mathcal S_{-i}(v_{-i}))$. 

Consider a mechanism $\mathcal M$ in which each player $i$ has a dominant strategy $\mathcal{S}_i^\ast$.
Let $(f,P_1,\ldots,P_n)$ be the social choice function and the payment scheme they realize. In this case, the mechanism \emph{implements $(f,P_1,\ldots,P_n)$ in dominant strategies}. 
In particular, the fact that a mechanism implements $(f,P_1,\ldots,P_n)$ in dominant strategies implies that for every player $i$, for every $v_i,v_i'\in V_i$ and for every $v_{-i}\in V_{-i}$:
\begin{equation*} \label{eq-truthfulness}
v_i(f(v_i,v_{-i}))-P_i(v_i,v_{-i})\ge v_i(f(v_i',v_{-i}))-P_i(v_i',v_{-i})     
\end{equation*}

\subsubsection*{Obviously Strategy-Proof Mechanisms} 
We now define obviously dominant strategies. For that, we need to 
 delve even further into the nuances of mechanisms and behaviors. 
We begin by defining two properties of behaviors. 

A vertex $u$ of a protocol is \emph{attainable} given a behavior $B_i$ if there exists a behavior profile of the other players, $B_{-i}\in \mathcal{B}_{-i}$, such that $u\in Path(B_i,B_{-i})$. For example, in the mechanism depicted in Figure \ref{fig-prems}, vertex $N_2$ is attainable given the behavior $(N_1:"1")$ of the sunglasses-wearing duck, whereas vertex $N_3$ is not attainable given this behavior. 

In addition, given the paths of two behavior profiles $B$ and $B'$, we denote with $Path(B)\cap Path(B')$ all the nodes that belong to both $Path(B)$ and to $Path(B')$. For example,  given the mechanism described in Figure \ref{fig-prems}, the behavior profiles $B=\{B_s=(N_1:"2"),B_j=(N_2:"2",N_3:"2")\}$ and $B'=\{B_s'=(N_1:"2"),B_j'=(N_2:"1",N_3:"1")\}$ satisfy that  $Path(B)\cap Path(B')=\{N_1,N_3\}$. 
We can now define the notion of an obviously dominant behavior:
\begin{definition}\label{def-obvs-behavior}
Consider a mechanism $\mathcal M$, together with a behavior $B_i$ and a valuation $v_i$ of some player $i$.
Fix a vertex $u\in \mathcal N_i$ that is attainable given the behavior $B_i$. 
Behavior $B_i$ is an \emph{obviously dominant behavior for player $i$ at vertex $u$ given the valuation $v_i$} if for every behavior profiles $B_{-i}\in \mathcal B_{-i}$ and  $(B_1',\ldots,B_n')\in \mathcal B_1 \times \cdots \times \mathcal{B}_n$ such that:
\begin{enumerate}
    \item $u\in Path(B_1,\ldots,B_n)\cap Path(B_1',\ldots,B_n')$ \emph{and} 
    \item $B_i$ and $B_i'$ dictate sending different messages at vertex $u$.
\end{enumerate}
The following inequality holds:
$$
v_i(f_i(B_i,B_{-i})) - p_i(B_i,B_{-i}) \geq v_i(f_i(B_i',B_{-i}')) - p_i(B_i',B_{-i}')
$$
\end{definition}
If $B_i$ is obviously dominant given that the valuation is $v_i$ for all (relevant) vertices simultaneously, then $B_i$ is an \emph{obviously dominant behavior}. Formally:
\begin{definition}
Fix a behavior $B_i$ together with the subset of vertices in $\mathcal N_i$ that are attainable for it,  which we denote with $U_{B_i}$. Fix a valuation $v_i$ of player $i$.
The behavior $B_i$ is an \emph{obviously dominant behavior for player $i$ given the valuation $v_i$} if
it is an obviously dominant behavior for player $i$ given the valuation $v_i$ for every vertex $u\in U_{B_i}$. 
\end{definition} 
\begin{definition}
A strategy $\mathcal{S}_i$ of player $i$ is an \emph{obviously dominant strategy} if for every $v_i$, the behavior $\mathcal S_i(v_i)$ is an obviously dominant behavior for player $i$ given the valuation $v_i$. 
\end{definition}

Consider a  mechanism $\mathcal M$ together with the obviously dominant strategies $(\mathcal S_1,\ldots, \mathcal S_n)$ and let
$(f,P_1,\ldots,P_n)$ be the social choice function and the payment schemes  they realize. In this case, the mechanism \emph{implements $(f,P_1,\ldots,P_n)$ in obviously dominant strategies}, or put differently: the mechanism is \emph{obviously strategy-proof}.
An allocation rule $f$ that has payment schemes $P_1,\ldots,P_n$ such that there exists  an obviously strategy-proof  mechanism that implements them is \emph{OSP-implementable}. 
In our proofs, we will extensively use the following observation:
\begin{lemma}\label{lemma-bad-leaf-good-leaf}
Fix an obviously strategy-proof mechanism $\mathcal M$ with strategies $(\mathcal S_1,\ldots, \mathcal S_n)$
that realize an allocation rule and payment schemes $(f,P_1,\ldots,P_n):V_1\times\cdots \times V_n\to \mathcal T\times \mathbb R^n$.
Fix a player $i$, a vertex $u\in \mathcal N_i$
and
two valuation profiles $(v_i,v_{-i}),(v_i',v_{-i}')$ such that the following conditions hold simultaneously:
\begin{enumerate}
    \item $u\in Path(\mathcal S_i(v_i),\mathcal S_{-i}(v_{-i}))\cap Path(\mathcal S_i(v_i'), \mathcal S_{-i}(v_{-i}'))$. 
    \item $v_i(f(v_i,v_{-i}))-P_i(v_i,v_{-i})< 
v_i(f(v_i',v_{-i}'))-P_i(v_i',v_{-i}')$.
\end{enumerate}
Then,  the strategy  $\mathcal S_i$ dictates the same message for the valuations  $v_i$ and $v_i'$ at vertex $u$.   
\end{lemma}
The proof of Lemma \ref{lemma-bad-leaf-good-leaf} can be directly derived from the definition of obviously strategy-proof mechanisms. See Figure \ref{figure-pruning} for an illustration.
\begin{figure}[h]
    \centering
    \caption[LoF entry]{\footnotesize Below is an illustration of a vertex in a mechanism for which the conditions specified in Lemma \ref{lemma-bad-leaf-good-leaf} hold, but $\mathcal S_i(v_i)$ and $\mathcal S_i'(v_i)$ dictate different messages,  which we denote with $m$ and with $m'$ respectively. Denote $Leaf(\mathcal S_i(v_i'),\mathcal{S}_{-i}(v_{-i}'))$ with $\ell^G$ and  $Leaf(\mathcal S_i(v_i),\mathcal{S}_{-i}(v_{-i}))$ with $\ell^B$. By assumption, leaf $\ell^G$ is more profitable than leaf $\ell^B$ for bidder $i$ with valuation $v_i$. 
    As the figure illustrates, there is a contradiction, as the strategy $\mathcal S_i$ is not obviously dominant.
    Roughly speaking, it is because  it is not obvious for bidder $i$ given the valuation $v_i$ that she should send the message $m$ as  the strategy $\mathcal{S}_i$ dictates, rather than send message $m'$.}
    \label{figure-pruning}
    		\begin{tikzpicture}
			\node[shape=circle,draw=black,minimum size=0.8cm] (u) at (2,2) {$u$};
			\node[shape=circle,draw=black,minimum size=0.8cm] (v) at (0,0) {};
   \node [shape=circle, fill=magenta!30, draw=black, minimum size=0.8cm, label=below:{\small \begin{tabular}{c} $\mathcal S_i(v_i'),\mathcal S_{-i}(v_{-i}')$ \\$utility({v_i},\ell^{G})=HIGH$  \end{tabular}}] (l) at (-1,-1) {\small $\ell^G$};
			
			\node[shape=circle,draw=black,minimum size=0.8cm] (v') at (4,0) {}; 
			\node[shape=circle,draw=black,,fill=gray!20,minimum size=0.8cm,label=below:
   {\small \begin{tabular}{c} $\mathcal S_i(v_i),\mathcal S_{-i}(v_{-i})$ \\$utility({v_i},\ell^{B})=LOW$  \end{tabular}}] (l') at (5,-1) {\small$\ell^B$};
			\node[shape=circle,draw=black,minimum size=0.8cm] (n') at (3,-1) {}; 
			\node[shape=circle,draw=black,minimum size=0.8cm] (n) at (1,-1) {};

			\draw [->] (u) edge  node[sloped, above] {\footnotesize$m'$} (v);
			\draw [->] (u) edge  node[sloped, above] {\footnotesize $m$} (v');
			\draw [->] (v) edge[dotted]  node[sloped, above] {}  (l);
			\draw [->] (v') edge[dotted]  node[sloped, above] {} (l');
			\draw [->] (v) edge[dotted]  node[sloped, above] {} (n);
			\draw [->] (v') edge[dotted]  node[sloped, above] {} (n');
		\end{tikzpicture}
\end{figure}
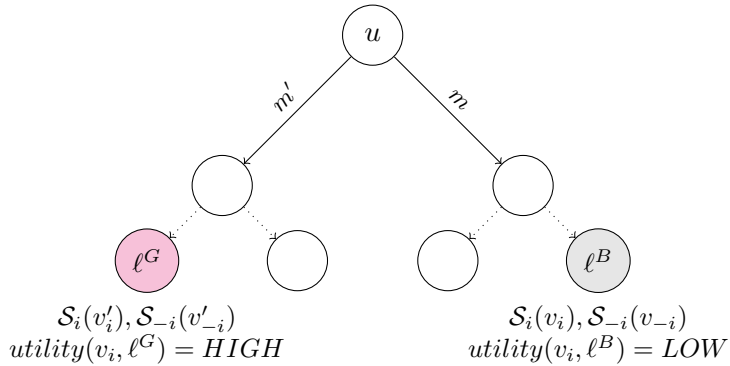
\subsubsection*{The Connection of Dominant-Strategy and Obviously Strategy-Proof Mechanisms}
It is easy to see that
every strategy that is
obviously dominant is also dominant, so every
obviously strategy-proof mechanism is a dominant-strategy mechanism. Interestingly, 
for mechanisms that are both sequential and  satisfy perfect information, the converse is also true:
\begin{proposition}\label{prop-obs}
Consider a  mechanism $\mathcal M$ together with the strategy profile
 $(\mathcal S_1,\ldots, \mathcal S_n)$ that implement an allocation rule and payments $(f,P_1,\ldots,P_n):V_1\times\cdots\times V_n\to \mathcal T \times \mathbb R^n$. Then, the mechanism $\mathcal M$ together with the strategies $(\mathcal S_1,\ldots, \mathcal S_n)$ implement $(f,P_1,\ldots,P_n)$ in dominant strategies if and only if they implement $(f,P_1,\ldots,P_n)$ in obviously dominant strategies.  
\end{proposition}
We defer the  proof of Proposition \ref{prop-obs} to Appendix \ref{subsec-prop-proof}. 
Proposition \ref{prop-obs}  depends crucially on the assumption that every node is associated with exactly one player that sends messages in it. 
Without this assumption, Proposition \ref{prop-obs} could have led to the erroneous conclusion that a sealed-bid implementation of a second-price auction, where all players simultaneously submit their values, is obviously strategy-proof.


\section{Impossibilities for Unknown Single-Minded Bidders} \label{sec-mua}
In this section, we consider
unknown single-minded bidders in multi-unit auctions (Subsection \ref{subsec-mua}) and combinatorial auctions (Subsection \ref{subsec-ca}). We show tight impossibilities for both settings under the standard assumptions of no-negative-transfers and individual rationality. The proofs are almost identical: The only difference between the proofs is the valuations we use. 


\subsection{The Multi-Unit Auction Setting} 
The definition of 
unknown single-minded bidders in a multi-unit auction is as follows. First, a valuation $v_i$ is \emph{single-minded} if it is parameterized with a quantity $q$ and a value $\alpha$ such that:
$$
v_i(s)=\begin{cases}
    \alpha \quad s\ge q, \\
    0 \quad \text{otherwise.}
\end{cases}
$$
{If the domain of valuations of a bidder contains solely single-minded valuations, we say that the bidder is \emph{single-minded}. If all the valuations are parameterized with the same quantity, then the bidder is \emph{known single-minded}, and otherwise she is \emph{unknown single-minded}.}

We remind that our impossibility result immediately implies that no obviously strategy-proof mechanism gives an approximation better than $\min\{m,n\}$ for a multi-unit auction with arbitrary monotone valuations (Corollary \ref{cor-mua-sm}). 
\label{subsec-mua} 
\subsubsection{Proof of Theorem \ref{thm-lb-mua}} \label{subsec-lb-mua}
Fix an auction with $m$ items and $n$ bidders. We assume that the domain $V_i$ of each bidder consists of single-minded valuations with values in  $\{0,1,\ldots,poly(m,n)\}$.

Assume towards a contradiction that 
there exists an obviously strategy-proof, individually rational, no-negative-transfers mechanism $\mathcal M$ together with strategy profile $\mathcal S=(\mathcal S_1,\ldots,\mathcal S_n)$
that implement an allocation rule and payment schemes  $(f,P_1,\ldots,P_n):V_1\times\cdots\times V_n\to\mathcal T\times \mathbb R^n$, where $f$ 
gives an approximation strictly better than $\min\{m,n\}$ to the optimal social welfare.  


For every player $i$, we define three valuations that will be of particular interest:  
$$
v_i^{all}(s)=\begin{cases}
k^4 \quad s=m, \\
0 \quad \text{otherwise.}
\end{cases} \\\quad 
v_i^{one}(s)= \begin{cases}
1 \quad s\ge 1, \\
0 \quad \text{otherwise.}
\end{cases}
 \quad
v_i^{ONE}(s)=\begin{cases}
k^2+1 \quad s\ge 1, \\
0 \quad \text{otherwise.}
\end{cases} 
$$
where $k=\max\{m,n\}$. Observe that given the valuation profile $(v_1^{one},\ldots,v_n^{one})$, there are at least two bidders such that $f$ assigns to them at least one item
because of its approximation guarantee.  
For convenience, assume those bidders are bidder $1$ and bidder $2$.\footnote{If those bidders would have been any other bidders, say bidders $3$ and $5$, the proof would be identical, except for the renaming of bidders.}

For the analysis of the mechanism, we define the following subsets of the domains of the valuations:
$
\mathcal{V}_1=\{v_1^{one},v_1^{ONE},v_1^{all}\}$,  $\mathcal{V}_2=\{v_2^{one},v_2^{ONE},v_2^{all}\}$ and for every player $i\ge 3$, we set $\mathcal{V}_i=\{v_i^{one}\}$. We denote $\mathcal{V}_1\times\cdots\times \mathcal{V}_n$ with $\mathcal V$.  

We will use the following lemma to show that at least one player $i$ at some vertex in the mechanism  has to send different messages for different valuations in $\mathcal V_i$. Formally:
\begin{lemma}\label{lemma-not-same-leaf}
There exist $(v_1,\ldots,v_n),(v_1',\ldots,v_n')\in \mathcal V_1\times\cdots\times \mathcal V_n$ such that the behaviors 
$(\mathcal S_1(v_1),\linebreak \ldots,\mathcal S_n(v_n))$ and $(\mathcal S_1(v_1'),\ldots,\mathcal S_n(v_n'))$ 
do not reach the same leaf in the mechanism. 
\end{lemma}
\begin{proof}
Consider the valuation profiles $v=(v_1^{all},v_2^{one},\ldots,v_n^{one})$ and $v'=(v_1^{one},\allowbreak v_2^{all},v_3^{one},\ldots,v_n^{one})$.  Given $v$, the only allocation that gives an approximation 
of $\min\{m,n\}$ to the welfare is allocating all items to player $1$, so $f(v)=(m,0,\ldots,0)$. Similarly, given the valuation profile $v'$, the only allocation rule that gives a $\min\{m,n\}$ approximation is allocating all the items to player $2$. Thus, the behavior profiles $\mathcal S(v)$ and $\mathcal S(v')$ reach different leaves. 
\end{proof}

We now use Lemma \ref{lemma-not-same-leaf}: Let $u$ be the first vertex in the protocol such that there exist $(v_1,\ldots,v_n),(v_1',\ldots,v_n')\in \mathcal V$ where the behavior profiles $(\mathcal{S}_1(v_1),\ldots,\mathcal{S}_n(v_n))$ and $(\mathcal{S}_1(v_1'),\ldots,\mathcal{S}_n(v_n'))$ diverge. Note that $u\in Path(\mathcal{S}_1(v_1),\ldots,\mathcal{S}_n(v_n))\cap Path(\mathcal{S}_1(v_1'),\ldots,\mathcal{S}_n(v_n'))$.
We remind that each vertex is associated with only one player that sends messages in it. The player that is associated with vertex $u$ is necessarily either player $1$ or player $2$ because every  player  $i\ge 3$    has  only one valuation in $\mathcal V_i$, so the set  $\{\mathcal S_i(v_i)\}_{v_i\in \mathcal V_i}$ is a singleton.  

We assume without loss of generality that vertex $u$ is associated with player $1$, meaning that    there exist $v_1,v_1'\in \mathcal{V}_1$ such that $\mathcal S_1(v_1)$ and $\mathcal S_1(v_1')$ dictate different messages at vertex $u$. We remind that $\mathcal{V}_1=\{v_1^{one},v_1^{ONE},v_1^{all}\}$, so
the  following claims jointly imply a contradiction, completing the proof:
\begin{claim}\label{claim-oneone-same}
    The strategy $\mathcal S_1$ dictates the same message at vertex $u$ for the valuations $v_1^{one}$ and $v_1^{ONE}$. 
\end{claim}
\begin{claim}\label{claim-one-all-same}
        The strategy $\mathcal S_1$ dictates the same message at vertex $u$ for the valuations $v_1^{ONE}$ and $v_1^{all}$.
\end{claim}

In the proofs of Claims \ref{claim-oneone-same} and Claim \ref{claim-one-all-same}, we use the following lemma, which is a collection of observations about the allocation and the payment scheme of player $1$:    
\begin{lemma}\label{lemma-small-pay}
    The allocation rule $f$ and the payment scheme $P_1$ of bidder $1$ satisfy that:
    \begin{enumerate}
        \item Given $(v_1^{one},v_{2}^{one},\ldots,v_{n}^{one})$, bidder $1$ wins at least one item and pays at most $1$.  \label{item-1}
        \item  Given $(v_1^{ONE},v_2^{all},v_3^{one},\ldots,v_n^{one})$, bidder $1$ gets the empty bundle and pays zero.   \label{item-2}
        \item Given $(v_1^{all},v_2^{one},\ldots,v_n^{one})$, bidder $1$ wins $m$ items and pays at most $k^2$. \label{item-3}  
    \end{enumerate}
\end{lemma}
The lemma is a direct consequence of the properties of the mechanism. 
We use 
it now and defer the proof to Appendix \ref{subsec-small-pay-proof}.
\begin{proof}[Proof of Claim \ref{claim-oneone-same}]
    Note that by Lemma \ref{lemma-small-pay} part \ref{item-1}, $f(v_1^{one},\ldots,v_n^{one})$ allocates at least one item to player $1$ and $P_1(v_1^{one},\ldots,v_n^{one})\le 1$. Therefore:
\begin{equation}\label{eq-good-leaf1}
 v_1^{ONE}(f(v_1^{one},\ldots,v_n^{one}))-P_1(v_1^{one},\ldots,v_n^{one})\ge k^2   
\end{equation}
 In contrast, by part \ref{item-2} of Lemma \ref{lemma-small-pay},   $f(v_1^{ONE},v_2^{all},v_3^{one},\ldots,v_n^{one})$ allocates no items to player $1$ and $P_1(v_1^{ONE},v_2^{all},v_3^{one},\ldots,v_n^{one})=0$, so:
 \begin{equation}\label{eq-bad-leaf1}
 v_1^{ONE}(f(v_1^{ONE},v_2^{all},v_3^{one},\ldots,v_n^{one}))-P_1(v_1^{ONE},v_2^{all},v_3^{one},\ldots,v_n^{one})= 0   
\end{equation}
Combining inequalities (\ref{eq-good-leaf1}) and (\ref{eq-bad-leaf1}) gives:
\begin{align*}
  v_1^{ONE}(f(v_1^{ONE},v_2^{all},v_3^{one},\ldots,v_n^{one}))-&P_1(v_1^{ONE},v_2^{all},v_3^{one},\ldots,v_n^{one})< \\
  &v_1^{ONE}(f(v_1^{one},\ldots,v_n^{one}))-P_1(v_1^{one},\ldots,v_n^{one})  
\end{align*}
We remind that vertex $u$ belongs in $Path(\mathcal S_1(v_1^{one}),\mathcal S_2(v_2^{one}),\ldots,\mathcal S_n(v_n^{one}))$ and also in
$Path(\mathcal{S}_1(v_1^{ONE}),\allowbreak\mathcal{S}_2(v_2^{all}), \mathcal{S}_3(v_3^{one}),\ldots,\mathcal{S}_n(v_n^{one}))$. Therefore, Lemma \ref{lemma-bad-leaf-good-leaf} gives that the strategy $\mathcal S_1$ dictates the same message for  $v_1^{one}$ and $v_1^{ONE}$ at vertex $u$. See Figure \ref{subfig-1} for an illustration.
\end{proof}


\begin{figure}[h]
    \centering
    \caption{\footnotesize Illustrations of the tree rooted at vertex $u$ of the mechanism $\mathcal M$. Figure \ref{subfig-1} depicts the case where strategy $\mathcal S_1$ dictates different messages for $v_1^{one}$ and $v_1^{ONE}$. Figure \ref{subfig-12} describes the case where  $\mathcal S_1$ dictates different messages for $v_1^{all}$ and $v_1^{ONE}$. 
    Both of them demonstrate that if bidder $1$ sends different messages for either $v_1^{one}$ and $v_1^{ONE}$ or for $v_1^{ONE}$ and $v_1^{all}$, then the strategy $\mathcal S_1$ is not obviously dominant.  
    }
\begin{subfigure}[b]{0.45\textwidth}
		
    \centering
        \setlength{\belowcaptionskip}{12pt}

    		\begin{tikzpicture}[scale=0.8]
			\node[shape=circle,draw=black,minimum size=0.8cm] (u) at (2,2) {$u$};
			\node[shape=circle,draw=black,minimum size=0.8cm] (v) at (0,0) {};
			\node[shape=circle,fill=magenta!30,draw=black,minimum size=0.8cm,label=below:{\small $(s_1\ge 1,p_1\le 1)$}] (l) at (-1,-1) {\small$\ell$};
			\node[shape=circle,draw=black,minimum size=0.8cm] (v') at (4,0) {}; 
			\node[shape=circle,draw=black,,fill=gray!40,minimum size=0.8cm,label=below:{\small $(s_1'=0,p_1'=0)$}] (l') at (5,-1) {\small$\ell'$};
			\node[shape=circle,draw=black,minimum size=0.8cm] (n') at (3,-1) {}; 
			\node[shape=circle,draw=black,minimum size=0.8cm] (n) at (1,-1) {};

			\draw [->] (u) edge  node[sloped, above] {\scriptsize$\mathcal S_1(v_1^{one})$} (v);
			\draw [->] (u) edge  node[sloped, above] {\scriptsize $\mathcal S_1(v_1^{ONE})$} (v');
			\draw [->] (v) edge[dotted]  node[sloped, above] {}  (l);
			\draw [->] (v') edge[dotted]  node[sloped, above] {} (l');
			\draw [->] (v) edge[dotted]  node[sloped, above] {} (n);
			\draw [->] (v') edge[dotted]  node[sloped, above] {} (n');
		\end{tikzpicture}
  \caption{\footnotesize Let us denote the leaves that $(\mathcal S_1(v_1^{one}),\ldots,\mathcal S_n(v_n^{one}))$
and $(\mathcal S_1(v_1^{ONE}),\allowbreak S_2(v_2^{all}),\allowbreak S_3(v_3^{one})\ldots,\mathcal S_n(v_n^{one}))$ reach
with $\ell$ and  $\ell'$ respectively. The outcome of leaf $\ell$ is more profitable than the outcome of leaf $\ell'$ for bidder $i$ with valuation $v_i^{ONE}$, which is why $\mathcal S_1$ is not obviously dominant in this case.} \label{subfig-1} 
\end{subfigure}
      \hspace{12pt} 
\begin{subfigure}[b]{0.45\textwidth}
        \setlength{\belowcaptionskip}{12pt}
    
    \centering
    		\begin{tikzpicture}[scale=0.8]
			\node[shape=circle,draw=black,minimum size=0.8cm] (u) at (2,2) {$u$};
			\node[shape=circle,draw=black,minimum size=0.8cm] (v) at (0,0) {};
			\node[shape=circle,fill=magenta!30,draw=black,minimum size=0.8cm,label=below:{\small $(s_1=m,p_1\le k^2)$}] (l) at (-1,-1) {\small$\ell$};
			\node[shape=circle,draw=black,minimum size=0.8cm] (v') at (4,0) {}; 
			\node[shape=circle,draw=black,,fill=gray!40,minimum size=0.8cm,label=below:{\small $(s_1'=0,p_1'=0)$}] (l') at (5,-1) {\small$\ell'$};
			\node[shape=circle,draw=black,minimum size=0.8cm] (n') at (3,-1) {}; 
			\node[shape=circle,draw=black,minimum size=0.8cm] (n) at (1,-1) {};

			\draw [->] (u) edge  node[sloped, above] {\scriptsize$\mathcal S_1(v_1^{all})$} (v);
			\draw [->] (u) edge  node[sloped, above] {\scriptsize $\mathcal S_1(v_1^{ONE})$} (v');
			\draw [->] (v) edge[dotted]  node[sloped, above] {}  (l);
			\draw [->] (v') edge[dotted]  node[sloped, above] {} (l');
			\draw [->] (v) edge[dotted]  node[sloped, above] {} (n);
			\draw [->] (v') edge[dotted]  node[sloped, above] {} (n');
		\end{tikzpicture}
  \caption{\footnotesize We denote the leaves that $(\mathcal S_1(v_1^{all}),\mathcal S_2(v_2^{one}),\ldots,\mathcal S_n(v_n^{one}))$
and $(\mathcal S_1(v_1^{ONE}),\allowbreak S_2(v_2^{all}),\allowbreak S_3(v_3^{one})\ldots,\mathcal S_n(v_n^{one}))$ reach
with $\ell$ and  $\ell'$ respectively. The outcome of leaf $\ell$ is more profitable than the outcome of leaf $\ell'$ for bidder $i$ with valuation $v_i^{ONE}$, which is why $\mathcal S_1$ is not obviously dominant in this case.} \label{subfig-12}
\end{subfigure}
\end{figure}
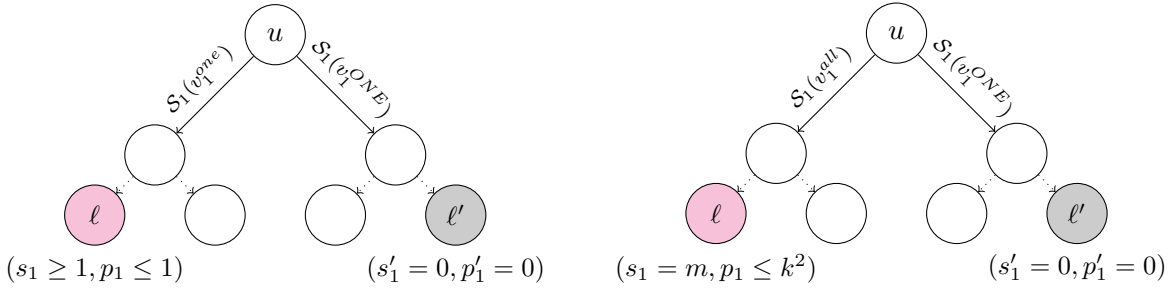
\begin{proof}[Proof of Claim \ref{claim-one-all-same}]
Following the same approach as in the proof of Claim \ref{claim-oneone-same}:
\begin{align*}
v_1^{ONE}(f(v_1^{all},v_2^{one},\ldots,v_n^{one}))-P_1(v_1^{all},v_2^{one},\ldots,v_n^{one}) &\ge k^2+1-k^2 \quad \text{(Lemma \ref{lemma-small-pay}.\ref{item-3})}
\\
& >0\\
&=v_1^{ONE}(f(v_1^{ONE},v_2^{all},v_3^{one},\ldots,v_n^{one}))  \\ 
&-P_1(v_1^{ONE},v_2^{all},v_3^{one},\ldots,v_n^{one}) \quad \text{(Lemma \ref{lemma-small-pay}.\ref{item-2})} 
\end{align*}
Since by assumption vertex $u$ belongs in $Path(\mathcal S_1(v_1^{all}),\mathcal S_2(v_2^{one}),\ldots,\mathcal S_n(v_n^{one}))$ and in
$Path((\mathcal{S}_1(v_1^{ONE}),\allowbreak\mathcal{S}_2(v_2^{all}), \mathcal{S}_3(v_3^{one}),\ldots,\mathcal{S}_n(v_n^{one})))$. By Lemma \ref{lemma-bad-leaf-good-leaf}, the strategy $\mathcal S_1$ dictates the same message for the valuations $v_1^{ONE}$ and $v_1^{all}$ at vertex $u$. See Figure \ref{subfig-12} for an illustration.    
\end{proof}

\subsection{The Combinatorial Auction Setting} \label{subsec-ca}
In this section, we consider unknown single-minded bidders in a combinatorial auction.   
The definition of 
unknown single-minded bidders is as follows:  a valuation $v_i$ is \emph{single-minded} if it is parameterized with a bundle $S^\ast$ and a value $\alpha$ such that:
$$
v_i(S)=\begin{cases}
    \alpha \quad S^\ast\subseteq  S, \\
    0 \quad \text{otherwise.}
\end{cases}
$$
Similarly to the definition of unknown single-minded bidders in a multi-unit auctions, 
bidder is \emph{single-minded} if all the valuations in her domain are single-minded. If all the valuations are parameterized with the same quantity, then the bidder is \emph{known single-minded} and otherwise, she is \emph{unknown single-minded}.
We now prove a tight impossibility for unknown single-minded bidders:

\begin{proof}[Proof of Theorem \ref{cor-ca-sm}]    
Fix an auction with a set of items  $M$ and $n$ bidders. 
Denote the elements in the set of items with $\{e_1,\ldots,e_m\}$. 
We assume that the domain $V_i$ of each bidder consists of single-minded valuations with values in  $\{0,1,\ldots,poly(m,n)\}$.

Assume towards a contradiction that 
there exists an obviously strategy-proof, individually rational, no-negative-transfers mechanism $\mathcal M$ together with strategy profile $\mathcal S=(\mathcal S_1,\ldots,\mathcal S_n)$
that implement an allocation rule and payment schemes  $(f,P_1,\ldots,P_n):V_1\times\cdots\times V_n\to\mathcal T\times \mathbb R^n$, where $f$ 
gives an approximation strictly better than $\min\{m,n\}$ to the optimal social welfare.  



We define several valuations. For each player $i\le m$, we define:
$$
v_i^{one}(S)=\begin{cases}
1 \quad e_i\in S, \\
0 \quad \text{otherwise.}
\end{cases}
$$
For each player $i>m$, we set:
$$
v_i^{one}(S)=\begin{cases}
1 \quad e_m\in S, \\
0 \quad \text{otherwise.}
\end{cases}
$$
Observe that given the valuation profile $(v_1^{one},\ldots,v_n^{one})$, there are at least two bidders that win a valuable item because of the approximation guarantee of $f$.
For convenience, assume those bidders are bidder $1$ and bidder $2$ who win items $e_1$ and $e_2$ respectively.
We now define two additional valuations for bidders $1$ and $2$:
$$
v_1^{all}(S)=\begin{cases}
k^4 \quad S=M, \\
0 \quad \text{otherwise.}
\end{cases} \\\quad 
v_1^{ONE}(S)= \begin{cases}
k^2+1 \quad e_1\in S, \\
0 \quad \text{otherwise.}
\end{cases}
$$
$$
v_2^{all}(S)=\begin{cases}
k^4 \quad S=M, \\
0 \quad \text{otherwise.}
\end{cases} \\\quad 
v_2^{ONE}(S)= \begin{cases}
k^2+1 \quad e_2\in S, \\
0 \quad \text{otherwise.}
\end{cases}
$$
where $k=\max\{m,n\}$. 
For the analysis of the mechanism, we define the following subsets of the domains of the valuations:
$
\mathcal{V}_1=\{v_1^{one},v_1^{ONE},v_1^{all}\}$,  $\mathcal{V}_2=\{v_2^{one},v_2^{ONE},v_2^{all}\}$ and for every player $i\ge 3$, we set $\mathcal{V}_i=\{v_i^{one}\}$. 
The rest of the proof is identical to the proof of Theorem \ref{thm-lb-mua}. 
\end{proof}

\section{Obviously Strategy-Proof Mechanisms 
for Additive  Valuations}\label{sec-additive-unit}
In this section, we analyze mechanisms for additive valuations. We remind that a valuation $v_i:2^{M}\to \mathbb R$ is \emph{additive} if $v_i(S)=\sum_{j\in S}v_i(\{j\})$.
We begin by showing that for a restricted class of additive valuations, there is an obviously strategy-proof mechanism (Subsection \ref{subsec-mech}). 
In Subsection \ref{subsec-impos-add},
we
show that under standard assumptions, no obviously strategy-proof mechanism gives an approximation better than $\min\{m,n\}$.  
The proof of impossibility for unit-demand valuations (Theorem \ref{thm-lb-unit})
 is very similar to the proof of Theorem \ref{thm-lb-add}. It can be found in Appendix \ref{appsec-unit}.

\subsection{Proof of Theorem \ref{thm-ub}: A Welfare-Maximizing Mechanism for a  (Restricted) Class of Additive Valuations} \label{subsec-mech}
Consider the following serial mechanism.
Player $1$ reports the items for which his value is $x_h$ and gets them at a price of $x_l$ per item. Next, player $2$ reports the remaining items for which his value is $x_h$,  gets them at a price of $x_l$, and the round procedure continues until no items or bidders are left.
If there are items left at the end of this round, we run a second serial round: this time, each bidder takes the remaining items for which his value is $x_l$. 

This mechanism maximizes the welfare because every item is allocated to a bidder that has maximum value for it: 
if an item has a bidder with value $x_h$ for it, by construction such a bidder wins it. If the highest value for the item is $x_l$, then it is allocated to a bidder at price $x_l$.

The mechanism is obviously strategy-proof  because  the utility of each bidder that takes $j$ high-valued items in the first round is fixed to be $j\cdot (x_h-x_l)$, no matter what the other bidders do or how many items with value $x_l$ he wins at the second round.
Therefore, 
answering the queries of the mechanism truthfully is an obviously dominant strategy for all bidders.

\subsection{Impossibility Result: Proof of Theorem \ref{thm-lb-add}}
 \label{subsec-impos-add}

 Fix an auction with a set of items  $M$ and $n$ bidders, where the   domain $V_i$ of each bidder consists of additive valuations with values in  $\{0,1,\ldots,poly(m,n)\}$. We denote the elements  of $M$ with $\{e_1,\ldots,e_m\}$.
 
Assume towards a contradiction that 
there exists an obviously strategy-proof, individually rational, no-negative-transfers mechanism $\mathcal M$ together with strategies $\mathcal S=(\mathcal S_1,\ldots,\mathcal S_n)$
that implement an allocation rule and payment schemes  $(f,P_1,\ldots,P_n):V_1\times\cdots\times V_n\to\mathcal T\times \mathbb R^n$, where $f$ 
gives an approximation strictly better than $\min\{m,n\}$ to the optimal social welfare.  

{We start by defining several valuations. Since these valuations are additive, we can fully describe them by specifying their value for each item.}  
For every player $i\le m$, we define:
$$
v_i^{e_i,one}(x)=\begin{cases}
1 \quad x=e_i, \\
0 \quad \text{otherwise.}
\end{cases}
$$
Whereas for every player $i>m$:
$$
v_i^{e_i,one}(x)=\begin{cases}
1 \quad x=e_m, \\
0 \quad \text{otherwise.}
\end{cases}
$$
 Observe that given the valuation profile $(v_1^{e_1,one},\ldots,v_n^{e_n,one})$, the approximation ratio of $f$
 guarantees that there are at least two bidders such that $f$ assigns to them a bundle that is valuable for them. For convenience, we assume that bidders $1$ and $2$ are allocated $e_1$ and $e_2$.\footnote{{Similarly to the proof of Theorem \ref{thm-lb-mua}, this is without loss of generality because if it would have been any other pair of bidders, the proof would have remained the same up to renaming of the bidders.}}
For the analysis, We define  additional valuations of players $1$ and $2$:
$$
v_1^{e_1,big}(x)= \begin{cases}
3k^4 \quad x=e_1, \\
0 \quad \text{otherwise.}
\end{cases}
v_1^{e_2,big}(x)=\begin{cases}3k^4 \quad x=e_2, \\
0, \quad \text{otherwise.}
\end{cases} 
v_1^{both}(x)=\begin{cases}
2k^2+2 \quad x=e_1, \\
2k^2 \quad x=e_2, \\
0, \quad \text{otherwise.}
\end{cases} 
$$ 
$$
v_2^{e_2,big}(x)= \begin{cases}
3k^4 \quad x=e_2, \\
0 \quad \text{otherwise.}
\end{cases}
v_2^{e_1,big}(x)=\begin{cases}3k^4 \quad x=e_1, \\
0, \quad \text{otherwise.}
\end{cases} 
v_2^{both}(x)=\begin{cases}
2k^2+2 \quad x=e_2, \\
2k^2 \quad x=e_1, \\
0, \quad \text{otherwise.}
\end{cases} 
$$ 
where $k=\max\{m,n\}$. We define the following subsets of the domains of the valuations:
$$
\mathcal V_1=\{v_1^{e_1,one},v_1^{e_1,big},v_1^{e_2,big},v_1^{both}\}, \quad
\mathcal V_2=\{v_2^{e_2,one},v_2^{e_2,big},v_2^{e_1,big},v_2^{both}\}
$$
{and for every player $i\ge 3$,  we set $\mathcal{V}_i=\{v_i^{e_i,one}\}$.\footnote{Note that for every player $i$, all valuations in $\mathcal V_i$ except for $v_1^{both}$ and $v_2^{both}$ have only one valuable item. 
It is for this reason that  the proofs of Theorems \ref{thm-lb-add}
and \ref{thm-lb-unit} {are so similar}.} 
We will use the following lemma to show that at least one player $i$ at some point  has to send  different messages for different valuations in $\mathcal V_i$. Formally: 
\begin{lemma}\label{lemma-not-same-leaf-add}
There exist $(v_1,\ldots,v_n),(v_1',\ldots,v_n')\in \mathcal V_1\times\cdots\times \mathcal V_n$ such that the behaviors 
$(\mathcal S_1(v_1),\linebreak \ldots,\mathcal S_n(v_n))$ and $(\mathcal S_1(v_1'),\ldots,\mathcal S_n(v_n'))$ 
do not reach the same leaf in the mechanism. 
\end{lemma}
\begin{proof}
Consider the valuation profiles $v=(v_1^{e_1,big},v_2^{e_2,one},\ldots,v_n^{e_n,one})$ and $v'=(v_1^{e_1,one},\allowbreak  v_2^{e_1,big},v_3^{e_3,one},\allowbreak\ldots,v_n^{e_n,one})$.  Given the valuation profile $v$, 
player $1$ wins item $e_1$ because of the approximation guarantee of the mechanism. Similarly, given the valuation profile $v'$, player $2$ wins item $e_1$. 
Thus, 
the behavior profiles $\mathcal S(v)$ and $\mathcal S(v')$ reach different leaves.
\end{proof}

We now use Lemma \ref{lemma-not-same-leaf-add}. Let $u$ be the first vertex in the protocol such that there exist $(v_1,\ldots,v_n),(v_1',\ldots,v_n')\in \mathcal V$ where the behavior profiles $(\mathcal{S}_1(v_1),\ldots,\mathcal{S}_n(v_n))$ and $(\mathcal{S}_1(v_1'),\ldots,\mathcal{S}_n(v_n'))$ diverge. Note that $u\in Path(\mathcal{S}_1(v_1),\ldots,\mathcal{S}_n(v_n))\cap Path(\mathcal{S}_1(v_1'),\ldots,\mathcal{S}_n(v_n'))$.
We remind that each vertex is associated with only one player that sends messages in it. The player that is associated with vertex $u$ is necessarily either player $1$ or player $2$ because every  player  $i\ge 3$    has  only one valuation in $\mathcal V_i$, so the set  $\{\mathcal S_i(v_i)\}_{v_i\in \mathcal V_i}$ is a singleton.  
We assume without loss of generality that vertex $u$ is associated with player $1$, meaning that    there exist $v_1,v_1'\in \mathcal{V}_1$ such that $\mathcal S_1(v_1)$ and $\mathcal S_1(v_1')$ dictate different messages at vertex $u$.  

Lemmas \ref{lemma-case-1} and \ref{lemma-case-2} examine two complementary cases, demonstrating that in each scenario, the strategy $\mathcal S_1$
dictates the same message at vertex $u$ for all the valuations in $\mathcal V_1$. This leads to a contradiction, completing the proof.
\begin{lemma}\label{lemma-case-1}
If $f(v_1^{both},v_2^{e_1,big},v_3^{e_3,one},\ldots,v_n^{e_n,one})$ outputs an allocation where bidder $1$ wins a bundle that contains $e_2$, 
then the strategy $\mathcal S_1$ dictates the same message at vertex $u$ for all the valuations in $\mathcal V_1$.   
\end{lemma}
\begin{lemma}\label{lemma-case-2}
    If $f(v_1^{both},v_2^{e_1,big},v_3^{e_3,one},\ldots,v_n^{e_n,one})$
    outputs an allocation where bidder $1$ wins a bundle not containing 
$e_2$, 
    then the strategy $\mathcal S_1$ dictates the same message at vertex $u$ for all the valuations in $\mathcal V_1$.   
\end{lemma}

It remains to prove Lemma \ref{lemma-case-1} and Lemma \ref{lemma-case-2}. 
The rest of the proof follows this outline. In Subsection \ref{subsubsec-allocpayments}, we state Lemma \ref{lemma-small-pay-add}, which consists of observations about the allocation function and the payment scheme of the mechanism. In Subsection \ref{subsubsec-e1-together}, we prove Claim \ref{claim-e1-same-add}  which shows that the obviously dominant strategy $\mathcal{S}_1$ always dictates the same message for the valuations  $v_1^{e_1,one}$ and $v_1^{e_1,big}$. Both of these facts will be used in the proofs of Lemma \ref{lemma-case-1} and Lemma \ref{lemma-case-2}, which are in Subsections \ref{subsubsec-case1} and \ref{subsubsec-case2} respectively.




\subsubsection{Observations about the Allocation and Payments of the Mechanism} \label{subsubsec-allocpayments}
We begin by deriving straightforward observations about the bundle that bidder $1$ wins and his payment:
\begin{lemma}\label{lemma-small-pay-add}
    The allocation rule $f$ and the payment scheme $P_1$ of bidder $1$ satisfy that:
    \begin{enumerate}
        \item Given $(v_1^{e_1,one},v_2^{e_2,one},\ldots,v_{n}^{e_n,one})$, bidder $1$ wins a bundle that contains $e_1$ and pays at most $1$.  \label{item-1-add}
        \item  Given  \label{item-2-add}$(v_1^{e_1,big},v_2^{e_1,big},v_3^{e_3,one},\ldots,v_n^{e_n,one})$, bidder $1$ either:
        \begin{enumerate*}[label=(\alph*)]
    \item gets a bundle not containing  $e_1$ and pays $0$ \emph{or}
    \item gets a bundle that contains $e_1$ and pays at least $2{k^3}+k^2$.
    \end{enumerate*}
        \item  Given  \label{item-3-add}$(v_1^{e_2,big},v_2^{e_2,big},v_3^{e_3,one},\ldots,v_n^{e_n,one})$, bidder $1$ either:
        \begin{enumerate*}[label=(\alph*)]
    \item gets a bundle not containing  $e_2$ and pays $0$ \emph{or}
    \item gets a bundle that contains $e_2$ and pays at least $2{k^3}+k^2$.
    \end{enumerate*}
    \item Given  \label{item-4-add}$(v_1^{both},v_2^{e_2,one},\ldots,v_n^{e_n,one})$,  bidder $1$ wins a bundle that contains $e_1$ and pays at most $4k^2+2$.
    \item Given $(v_1^{e_2,big},v_2^{e_2,one},\ldots,v_n^{e_n,one})$, bidder $1$ wins $e_2$ and pays at most $k^2$.  \label{item-5-add}
     \item Given $(v_1^{both},v_2^{e_1,big},v_3^{e_3,one},\ldots,v_n^{e_n,one})$, bidder $1$ does not win $e_1$. \label{item-new-add}
    \item If  $f(v_1^{both},v_2^{e_1,big},v_3^{e_3,one},\ldots,v_n^{e_n,one})$
    outputs an allocation where bidder $1$ wins $e_2$, then he pays at most $2k^2$. \label{item-6-add}
    \item If  $f(v_1^{both},v_2^{e_1,big},v_3^{e_3,one},\ldots,v_n^{e_n,one})$ \label{item-7-add}
    outputs an allocation where bidder $1$ wins a bundle not containing $e_2$,  then he pays zero.
    \end{enumerate} 
\end{lemma}
The lemma is a direct consequence of the properties of the mechanism. We prove it in Appendix \ref{appsubsec-proof-lemma-add-small-pay}.
\subsubsection{The Valuations $v_{1}^{e_1,one}$ and $v_{1}^{e_1,big}$ Send the Same Message} \label{subsubsec-e1-together}
We now use the observations about the  mechanism (Lemma \ref{lemma-small-pay-add}) to show that $\mathcal S_1$ dictates the same message for $v_{1}^{e_1,one}$ and $v_{1}^{e_1,big}$ at vertex $u$. {We use Claim \ref{claim-e1-same-add} to prove Lemma \ref{lemma-case-1} and Lemma \ref{lemma-case-2}. 
\begin{claim}\label{claim-e1-same-add}
    The strategy $\mathcal S_1$ dictates the same message at vertex $u$ for both $v_1^{e_1,one}$ and $v_1^{e_1,big}$.
\end{claim}
\begin{proof}
    Note that by Lemma \ref{lemma-small-pay-add} part \ref{item-1-add}, given $(v_1^{e_1,one},v_2^{e_2,one}\ldots,v_n^{e_n,one})$, bidder $1$ wins a bundle that contains $e_1$ and pays at most $1$. Therefore: 
    \begin{equation}\label{eq-prof-large-e1}
        v_1^{e_1,big}(f(v_1^{e_1,one},v_2^{e_2,one}\ldots,v_n^{e_n,one}))- P_1(v_1^{e_1,one},v_2^{e_2,one}\ldots,v_n^{e_n,one})\ge 3k^4-1
    \end{equation}
    In contrast, by Lemma \ref{lemma-small-pay-add} part \ref{item-2-add}:
    \begin{equation}\label{eq-prof-small-e1}
        v_1^{e_1,big}(f(v_1^{e_1,big},v_2^{e_1,big},v_3^{e_3,one},\ldots,v_n^{e_n,one}))- P_1(v_1^{e_1,big},v_2^{e_1,big},v_3^{e_3,one},\ldots,v_n^{e_n,one})\le 3k^4 - 2{k^3}-k^2
    \end{equation}
    Combining (\ref{eq-prof-large-e1}) and (\ref{eq-prof-small-e1}) gives:
    \begin{multline*}
        v_1^{e_1,big}(f(v_1^{e_1,one},v_2^{e_2,one}\ldots,v_n^{e_n,one}))- P_1(v_1^{e_1,one},v_2^{e_2,one}\ldots,v_n^{e_n,one})> \\ v_1^{e_1,big}(f(v_1^{e_1,big},v_2^{e_1,big},v_3^{e_3,one},\ldots,v_n^{e_n,one}))- P_1(v_1^{e_1,big},v_2^{e_1,big},v_3^{e_3,one},\ldots,v_n^{e_n,one})
    \end{multline*}
   We remind that vertex $u$ belongs in $Path(\mathcal S_1(v_1^{e_1,one}),\mathcal S_2(v_2^{e_2,one}),\ldots,\mathcal S_n(v_n^{e_n,one}))$ and  in
$Path(\mathcal{S}_1(v_1^{e_1,big}),\allowbreak\mathcal{S}_2(v_2^{e_1,big}), \mathcal{S}_3(v_3^{e_3,one}),\ldots,\mathcal{S}_n(v_n^{e_n,one}))$, so by Lemma \ref{lemma-bad-leaf-good-leaf} the strategy  $\mathcal S_1$ dictates the same message for   $v_1^{e_1,one}$ and $v_1^{e_1,big}$ at vertex $u$.
\end{proof}
\subsubsection{Proof of Lemma \ref{lemma-case-1}} \label{subsubsec-case1}
The proof is as follows. We first show that the strategy $\mathcal S_1$ always dictates the same message for $v_1^{both}$ and   
$v_1^{e_1,big}$ (Claim \ref{claim-sameM-both-large1}). We proceed by demonstrating that in the case under consideration in Lemma \ref{lemma-case-1}, the strategy
$\mathcal S_1$ dictates the same message for the valuations $v_1^{both}$ and $v_1^{e_2,big}$ (Claim \ref{claim-sameM-both-large2}).
We remind that $\mathcal V_1=\{v_1^{e_1,one},v_1^{e_1,big},v_1^{e_2,big},v_1^{both}\}$, so combining
Claims \ref{claim-e1-same-add}, \ref{claim-sameM-both-large1} and \ref{claim-sameM-both-large2} gives that the strategy $\mathcal S_1$ dictates the same message for all the valuations in $\mathcal V_1$, as needed. 

\begin{claim}\label{claim-sameM-both-large1}
The strategy $\mathcal S_1$ dictates the same message at vertex $u$ for both $v_1^{both}$  and $v_1^{e_1,big}$.
\end{claim}
\begin{proof}
Note that by Lemma \ref{lemma-small-pay-add} part \ref{item-4-add}, given $(v_1^{both},v_2^{e_2,one},\ldots,\allowbreak v_n^{e_n,one})$, bidder $1$ wins a bundle that contains $e_1$ and pays at most $4k^2+2$. Therefore:
\begin{equation}\label{eq-e1toboth-first}
    v_1^{e_1,big}(f(v_1^{both},v_2^{e_2,one},v_3^{e_3,one},\ldots,v_n^{e_n,one}))-P_1(v_1^{both},v_2^{e_2,one},\ldots,v_n^{e_n,one})\ge 3k^4-4k^2-2
\end{equation}
Whereas by Lemma \ref{lemma-small-pay-add} part \ref{item-2-add}:
\begin{equation}\label{eq-e1toboth2-first}
    v_1^{e_1,big}(f(v_1^{e_1,big},v_2^{e_1,big},v_3^{e_3,one},\ldots,v_n^{e_n,one}))-P_1(v_1^{e_1,big},v_2^{e_1,big},v_3^{e_3,one},\ldots,v_n^{e_n,one})\le 3k^4-2k^3-k^2
\end{equation}
Combining (\ref{eq-e1toboth-first}) and (\ref{eq-e1toboth2-first}) gives:
\begin{multline*}
    v_1^{e_1,big}(f(v_1^{e_1,big},v_2^{e_1,big},v_3^{e_3,one},\ldots,v_n^{e_n,one}))-P_1(v_1^{e_1,big},v_2^{e_1,big},v_3^{e_3,one},\ldots,v_n^{e_n,one})< \\ v_1^{e_1,big}(f(v_1^{both},v_2^{e_1,big},v_3^{e_3,one},\ldots,v_n^{e_n,one}))-P_1(v_1^{both},v_2^{e_1,big},v_3^{e_3,one},\ldots,v_n^{e_n,one})
\end{multline*}
We remind that by assumption vertex $u$ belongs in $Path(\mathcal S_1(v_1^{e_1,big}),\mathcal S_2(v_2^{e_1,big}),S_3(v_3^{e_3,one}),\ldots,\allowbreak\mathcal S_n(v_n^{e_n,one}))$ and also in
$Path(\mathcal{S}_1(v_1^{both}),\mathcal{S}_2(v_2^{e_2,one}), \mathcal{S}_3(v_3^{e_3,one}),\ldots,\mathcal{S}_n(v_n^{e_n,one}))$, so Lemma \ref{lemma-bad-leaf-good-leaf} gives that the strategy $\mathcal S_1$ dictates the same message for  $v_1^{e_1,big}$ and $v_1^{both}$ at vertex $u$.
\end{proof}

\begin{claim}\label{claim-sameM-both-large2}
If $f(v_1^{both},v_2^{e_1,big},v_3^{e_3,one},\ldots,v_n^{e_n,one})$ outputs an allocation where bidder $1$ wins a bundle that contains $e_2$, 
 then   the strategy $\mathcal S_1$ dictates the same message at vertex $u$ for both $v_1^{both}$  and $v_1^{e_2,big}$.
\end{claim}
\begin{proof}
    By Lemma \ref{lemma-small-pay-add} part \ref{item-6-add} and by assumption, given $(v_1^{both},v_2^{e_1,big},v_3^{e_3,one},\ldots,v_n^{e_n,one})$, bidder $1$  wins a bundle that contains $e_2$ and pays at most $2k^2$, so:
    \begin{equation}\label{eq-case1-claim2-1-first}
        v_1^{e_2,big}(f(v_1^{both},v_2^{e_1,big},v_3^{e_3,one},\ldots,v_n^{e_n,one}))-P_1(v_1^{both},v_2^{e_1,big},v_3^{e_3,one},\ldots,v_n^{e_n,one})\ge 3k^4-2k^2
    \end{equation}
    Whereas by Lemma \ref{lemma-small-pay-add} part \ref{item-3-add}:
     \begin{equation}\label{eq-case1-claim2-2-first}
        v_1^{e_2,big}(f(v_1^{e_2,big},v_2^{e_2,big},v_3^{e_3,one},\ldots,v_n^{e_n,one}))-P_1(v_1^{e_2,big},v_2^{e_2,big},v_3^{e_3,one},\ldots,v_n^{e_n,one})\le 3k^4-2k^3-k^2
    \end{equation}
Combining (\ref{eq-case1-claim2-1-first}) and (\ref{eq-case1-claim2-2-first}) gives:
\begin{multline*}
    v_1^{e_2,big}(f(v_1^{e_2,big},v_2^{e_2,big},v_3^{e_3,one},\ldots,v_n^{e_n,one}))-P_1(v_1^{e_2,big},v_2^{e_2,big},v_3^{e_3,one},\ldots,v_n^{e_n,one})< \\
    v_1^{e_2,big}(f(v_1^{both},v_2^{e_1,big},v_3^{e_3,one},\ldots,v_n^{e_n,one}))-P_1(v_1^{both},v_2^{e_1,big},v_3^{e_3,one},\ldots,v_n^{e_n,one})
\end{multline*}
Similarly to before, vertex $u$ is in $Path(\mathcal S_1(v_1^{e_2,big}),\mathcal S_2(v_2^{e_2,big}),\mathcal{S}_3(v_3^{e_3,one}),\ldots,\mathcal{S}_n(v_n^{e_n,one}))$ and also in $Path(\mathcal S_1(v_1^{both}),\allowbreak \mathcal S_2(v_2^{e_1,big}),\mathcal{S}_3(v_3^{e_3,one}),\ldots,\mathcal{S}_n(v_n^{e_n,one}))$, so applying Lemma \ref{lemma-bad-leaf-good-leaf} completes the proof. 
\end{proof}

\subsubsection{Proof of Lemma \ref{lemma-case-2}} \label{subsubsec-case2}
The proof of Lemma \ref{lemma-case-2} has the same structure as the proof of Lemma \ref{lemma-case-1}, except that we now analyze the case where bidder $1$ does not win $e_2$ in the instance  $(v_1^{both},v_2^{e_1,big},v_3^{e_3,one},\ldots,v_n^{e_n,one})$.
 We first show that for this case, the strategy $\mathcal S_1$ dictates the same message for $v_1^{both}$ and 
$v_1^{e_1,one}$ (Claim \ref{claim-sameM-both-small-e1-first}). We then show that it also implies that $\mathcal S_1$ dictates the same message for both $v_1^{both}$ and $v_1^{e_2,big}$ (Claim \ref{claim-sameM-both-large-e2}). Combining  Claims \ref{claim-sameM-both-small-e1-first} and \ref{claim-sameM-both-large-e2} with Claim \ref{claim-e1-same-add} gives that $\mathcal{S}_1$ sends the same message for all the valuations in $\mathcal V_1$, that completes the proof. 

\begin{claim}\label{claim-sameM-both-small-e1-first}
    If $f(v_1^{both},v_2^{e_1,big},v_3^{e_3,one},\ldots,v_n^{e_n,one})$
    outputs an allocation where bidder $1$ wins a bundle not containing $e_2$,   
then the strategy $\mathcal S_1$ dictates the same message at vertex $u$ for both $v_1^{both}$  and $v_1^{e_1,one}$.
\end{claim}
\begin{proof}
Note that by Lemma \ref{lemma-small-pay-add} part \ref{item-1-add}, given $(v_1^{e_1,one},v_2^{e_2,one},\ldots\allowbreak,v_n^{e_n,one})$, bidder $1$ wins a bundle that contains $e_1$ and pays at most $1$. Therefore:
\begin{equation}\label{eq-case2-e1to1}
    v_1^{both}(f(v_1^{e_1,one},v_2^{e_2,one},\ldots\allowbreak,v_n^{e_n,one}))-P_1(v_1^{e_1,one},v_2^{e_2,one},\ldots\allowbreak,v_n^{e_n,one})\ge 2k^2+1
\end{equation}
We now analyze the output of the mechanism for the valuation profile $(v_1^{both},v_2^{e_1,big},v_3^{e_3,one},\ldots,v_n^{e_n,one})$. By assumption and by Lemma \ref{lemma-small-pay-add} part \ref{item-new-add}, the function $f$ outputs for this instance
an allocation where bidder $1$ wins neither item $e_1$ nor item $e_2$. Therefore,  $v_1^{both}(f(v_1^{both},v_2^{e_1,big},v_3^{e_3,one},\ldots,v_n^{e_n,one}))=0$. Combining with 
 Lemma \ref{lemma-small-pay-add} part \ref{item-7-add}, we get that: 
\begin{equation}\label{eq-bothnothing}
    v_1^{both}(f(v_1^{both},v_2^{e_1,big},v_3^{e_3,one},\ldots,v_n^{e_n,one}))-P_1(v_1^{both},v_2^{e_1,big},v_3^{e_3,one},\ldots,v_n^{e_n,one})=0
\end{equation}
Combining (\ref{eq-case2-e1to1}) and (\ref{eq-bothnothing}) gives:
\begin{multline*}
    v_1^{both}(f(v_1^{both},v_2^{e_1,big},v_3^{e_3,one},\ldots,v_n^{e_n,one}))-P_1(v_1^{both},v_2^{e_1,big},v_3^{e_3,one},\ldots,v_n^{e_n,one})< \\ v_1^{both}(f(v_1^{e_1,one},v_2^{e_2,one},\ldots\allowbreak,v_n^{e_n,one})-P_1(v_1^{e_1,one},v_2^{e_2,one},\ldots\allowbreak,v_n^{e_n,one})
\end{multline*}
We remind that vertex $u$ belongs in $Path(\mathcal S_1(v_1^{e_1,one}),\mathcal S_2(v_2^{e_2,one}),\ldots,\mathcal S_n(v_n^{e_n,one}))$ and also belongs in
$Path(\mathcal{S}_1(v_1^{both}),\mathcal{S}_2(v_2^{e_1,big}), \mathcal{S}_3(v_3^{e_3,one}),\ldots,\mathcal{S}_n(v_n^{e_n,one}))$, so the proof is finished by applying Lemma \ref{lemma-bad-leaf-good-leaf}.
\end{proof}

\begin{claim}\label{claim-sameM-both-large-e2}
    If $f(v_1^{both},v_2^{e_1,big},v_3^{e_3,one},\ldots,v_n^{e_n,one})$
    outputs an allocation where bidder $1$  wins a bundle not containing $e_2$, 
then      the strategy $\mathcal S_1$ dictates the same message at vertex $u$ for both $v_1^{both}$  and $v_1^{e_2,big}$.
\end{claim}
\begin{proof}
Due to the same reasons specified in the proof of Claim \ref{claim-sameM-both-small-e1-first}:
    \begin{equation}\label{eq-case1-claim2-1}
        v_1^{both}(f(v_1^{both},v_2^{e_1,big},v_3^{e_3,one},\ldots,v_n^{e_n,one}))-P_1(v_1^{both},v_2^{e_1,big},v_3^{e_3,one},\ldots,v_n^{e_n,one})=0
    \end{equation}
    Whereas by Lemma \ref{lemma-small-pay-add} part \ref{item-5-add}:
     \begin{equation}\label{eq-case1-claim2-2}
        v_1^{both}(f(v_1^{e_2,big},v_2^{e_2,one},\ldots,v_n^{e_n,one}))-P_1(v_1^{e_2,big},v_2^{e_2,one},\ldots,v_n^{e_n,one})\ge 2k^2-k^2=k^2
    \end{equation}
Combining (\ref{eq-case1-claim2-1}) and (\ref{eq-case1-claim2-2}) gives:
\begin{multline*}
v_1^{both}(f(v_1^{both},v_2^{e_1,big},v_3^{e_3,one},\ldots,v_n^{e_n,one}))-P_1(v_1^{both},v_2^{e_1,big},v_3^{e_3,one},\ldots,v_n^{e_n,one})< \\
    v_1^{both}(f(v_1^{e_2,big},v_2^{e_2,one},\ldots,v_n^{e_n,one}))-P_1(v_1^{e_2,big},v_2^{e_2,one},\ldots,v_n^{e_n,one})
\end{multline*}
Similarly to before, vertex $u$ is in $Path(\mathcal S_1(v_1^{both}),\mathcal S_2(v_2^{e_1,big}),\mathcal{S}_3(v_3^{e_3,one}),\ldots,\mathcal{S}_n(v_n^{e_n,one}))$ and also in $Path(\mathcal S_1(v_1^{e_2,big}),\allowbreak \mathcal S_2(v_2^{e_2,one}),\ldots,\mathcal{S}_n(v_n^{e_n,one}))$, so applying Lemma \ref{lemma-bad-leaf-good-leaf} completes the proof. 
\end{proof}


\section*{Acknowledgements}
We would like to thank Shahar Dobzinski for  numerous helpful discussions, valuable advice, and feedback on earlier drafts of this paper. 
Additionally, we wish to thank Ariel Shaulker and the anonymous reviewers for their helpful comments regarding the presentation of the paper. 

\bibliographystyle{alpha}
\bibliography{main-arxiv}

\appendix

\section{An Impossibility Result for Unit-Demand Bidders: Proof of Theorem \ref{thm-lb-unit} }\label{appsec-unit}
In this section, we prove an impossibility result for unit-demand bidders. 
{We remind that a valuation $v_i$ is \emph{unit-demand} if $v_i(S)=\max_{j\in S}v_i(\{j\})$.} This proof is extremely similar to the proof of the Theorem \ref{thm-lb-add} in Section \ref{sec-additive-unit}. 
We write it for the sake of completeness. 
We note that the proofs are similar because most of the valuations in the analysis of Theorems \ref{thm-lb-unit} and Theorem \ref{thm-lb-add}  are both unit-demand and additive.

 Fix an auction with a set of items  $M$ and $n$ bidders, where the   domain $V_i$ of each bidder has  unit demand valuations with values in  $\{0,1,\ldots,poly(m,n)\}$. We denote the elements  of $M$ with $\{e_1,\ldots,e_m\}$.

Assume towards a contradiction that 
there exists an obviously strategy-proof, individually rational, no-negative-transfers mechanism $\mathcal M$ together with strategies $\mathcal S=(\mathcal S_1,\ldots,\mathcal S_n)$
that implement an allocation rule and payment schemes  $(f,P_1,\ldots,P_n):V_1\times\cdots\times V_n\to\mathcal T\times \mathbb R^n$, where $f$ 
gives an approximation strictly better than $\min\{m,n\}$ to the optimal social welfare.  

{We begin by defining several valuations that will be useful when analyzing the mechanism. We take advantage of the fact that they are unit-demand, so they can be fully described by specifying the value for each item separately.   
For every player $i\le m$ and for every item $x$, we define:
$$
v_i^{e_i,one}(x)=\begin{cases}
1 \quad x=e_i, \\
0 \quad \text{otherwise.}
\end{cases}
$$
Whereas for every player $i>m$:
$$
v_i^{e_i,one}(x)=\begin{cases}
1 \quad x=e_m, \\
0 \quad \text{otherwise.}
\end{cases}
$$
 Observe that given the valuation profile $(v_1^{e_1,one},\ldots,v_n^{e_n,one})$,
 the approximation ratio of $f$
 guarantees that there are at least two bidders such that $f$ assigns to them a bundle that is valuable for them. Similarly to the proof of Theorem \ref{thm-lb-add}, we assume for convenience and without loss of generality  
that  bidders  $1$ and $2$ are allocated $e_1$ and $e_2$.
We now define  additional valuations of players $1$ and $2$:
$$
v_1^{e_1,big}(x)= \begin{cases}
3k^4 \quad x=e_1, \\
0 \quad \text{otherwise.}
\end{cases}
v_1^{e_2,big}(x)=\begin{cases}3k^4 \quad x=e_2, \\
0, \quad \text{otherwise.}
\end{cases} 
v_1^{both}(x)=\begin{cases}
2k^2+2 \quad x=e_1, \\
2k^2 \quad x=e_2, \\
0, \quad \text{otherwise.}
\end{cases} 
$$ 
$$
v_2^{e_2,big}(x)= \begin{cases}
3k^4 \quad x=e_2, \\
0 \quad \text{otherwise.}
\end{cases}
v_2^{e_1,big}(x)=\begin{cases}3k^4 \quad x=e_1, \\
0, \quad \text{otherwise.}
\end{cases} 
v_2^{both}(x)=\begin{cases}
2k^2+2 \quad x=e_2, \\
2k^2 \quad x=e_1, \\
0, \quad \text{otherwise.}
\end{cases} 
$$ 
where $k=\max\{m,n\}$. We define the following subsets of the domains of the valuations:
$$
\mathcal V_1=\{v_1^{e_1,one},v_1^{e_1,big},v_1^{e_2,big},v_1^{both}\}, \quad
\mathcal V_2=\{v_2^{e_2,one},v_2^{e_2,big},v_2^{e_1,big},v_2^{both}\}
$$
and for every player $i\ge 3$, we set  $\mathcal{V}_i=\{v_i^{e_i,one}\}$.    
We will use the following lemma to show that at least one player $i$ at some point in the mechanism  has to send different messages for different valuations in $\mathcal V_i$. Formally: 
\begin{lemma}\label{lemma-not-same-leaf-unit}
There exist $(v_1,\ldots,v_n),(v_1',\ldots,v_n')\in \mathcal V_1\times\cdots\times \mathcal V_n$ such that the behaviors 
$(\mathcal S_1(v_1),\linebreak \ldots,\mathcal S_n(v_n))$ and $(\mathcal S_1(v_1'),\ldots,\mathcal S_n(v_n'))$ 
do not reach the same leaf in the mechanism. 
\end{lemma}
The proof of Lemma \ref{lemma-not-same-leaf-unit} is identical to the proof of Lemma \ref{lemma-not-same-leaf-add}. We now use it. Let $u$ be the first vertex in the protocol such that there exist $(v_1,\ldots,v_n),(v_1',\ldots,v_n')\in \mathcal V$ where the behavior profiles $(\mathcal{S}_1(v_1),\ldots,\mathcal{S}_n(v_n))$ and $(\mathcal{S}_1(v_1'),\ldots,\mathcal{S}_n(v_n'))$ diverge. Note that $u\in Path(\mathcal{S}_1(v_1),\ldots,\mathcal{S}_n(v_n))\cap Path(\mathcal{S}_1(v_1'),\ldots,\mathcal{S}_n(v_n'))$.
We remind that each vertex is associated with only one player that sends messages in it. The player that is associated with vertex $u$ is necessarily either player $1$ or player $2$ because every  player  $i\ge 3$    has  only one valuation in $\mathcal V_i$, so the set  $\{\mathcal S_i(v_i)\}_{v_i\in \mathcal V_i}$ is a singleton.  
We assume without loss of generality that vertex $u$ is associated with player $1$, meaning that    there exist $v_1,v_1'\in \mathcal{V}_1$ such that $\mathcal S_1(v_1)$ and $\mathcal S_1(v_1')$ dictate different messages at vertex $u$.  

Lemmas \ref{lemma-case-1-unit} and \ref{lemma-case-2-unit} examine two complementary cases, demonstrating that in each scenario, the strategy $\mathcal S_1$
dictates the same message at vertex $u$ for all the valuations in $\mathcal V_1$. This leads to a contradiction, completing the proof.
\begin{lemma}\label{lemma-case-1-unit}
If $f(v_1^{both},v_2^{e_1,big},v_3^{e_3,one},\ldots,v_n^{e_n,one})$ outputs an allocation where bidder $1$ wins a bundle that contains $e_2$, 
then the strategy $\mathcal S_1$ dictates the same message at vertex $u$ for all the valuations in $\mathcal V_1$.   
\end{lemma}
\begin{lemma}\label{lemma-case-2-unit}
    If $f(v_1^{both},v_2^{e_1,big},v_3^{e_3,one},\ldots,v_n^{e_n,one})$
    outputs an allocation where bidder $1$ wins a bundle not containing 
$e_2$, 
    then the strategy $\mathcal S_1$ dictates the same message at vertex $u$ for all the valuations in $\mathcal V_1$.   
\end{lemma}

It remains to prove Lemma \ref{lemma-case-1-unit} and Lemma \ref{lemma-case-2-unit}. Before we do so, we need two components. The first one is the following claim:
\begin{claim}\label{claim-e1-same-unit}
    The strategy $\mathcal S_1$ dictates the same message at vertex $u$ for both $v_1^{e_1,one}$ and $v_1^{e_1,big}$.
\end{claim} 
The proof of Claim \ref{claim-e1-same-unit} is identical to the proof of Claim \ref{claim-e1-same-add}. 
The second component comprises the following 
straightforward observations regarding the bundle that bidder $1$ wins and his payment: 
\begin{lemma}\label{lemma-small-pay-unit}
    The allocation rule $f$ and the payment scheme $P_1$ of bidder $1$ satisfy that:
    \begin{enumerate}
        \item Given $(v_1^{e_1,one},v_2^{e_2,one},\ldots,v_{n}^{e_n,one})$, bidder $1$ wins a bundle that contains $e_1$ and pays at most $1$.  \label{item-1-unit}
        \item  Given  \label{item-2-unit}$(v_1^{e_1,big},v_2^{e_1,big},v_3^{e_3,one},\ldots,v_n^{e_n,one})$, bidder $1$ either:
        \begin{enumerate*}[label=(\alph*)]
    \item gets a bundle not containing  $e_1$ and pays $0$ \emph{or}
    \item gets a bundle that contains $e_1$ and pays at least $2{k^3}+k^2$.
    \end{enumerate*}
        \item  Given  \label{item-3-unit}$(v_1^{e_2,big},v_2^{e_2,big},v_3^{e_3,one},\ldots,v_n^{e_n,one})$, bidder $1$ either:
        \begin{enumerate*}[label=(\alph*)]
    \item gets a bundle not containing $e_2$ and pays $0$ \emph{or}
    \item gets a bundle that contains $e_2$ and pays at least $2{k^3}+k^2$.
    \end{enumerate*}
    \item Given  \label{item-4-unit}$(v_1^{both},v_2^{e_2,one},\ldots,v_n^{e_n,one})$,  bidder $1$ wins a bundle that contains $e_1$ and pays at most $1$.
    \item Given $(v_1^{e_2,big},v_2^{e_2,one},\ldots,v_n^{e_n,one})$, bidder $1$ wins $e_2$ and pays at most $k^2$.  \label{item-5-unit}
    \item Given $(v_1^{both},v_2^{e_1,big},v_3^{e_3,one},\ldots,v_n^{e_n,one})$, bidder $1$ does not win $e_1$. \label{item-new-unit}
    \item If  $f(v_1^{both},v_2^{e_1,big},v_3^{e_3,one},\ldots,v_n^{e_n,one})$
    outputs an allocation where bidder $1$ wins item $e_2$, then he pays at most $2k^2$. \label{item-6-unit}
    \item If  $f(v_1^{both},v_2^{e_1,big},v_3^{e_3,one},\ldots,v_n^{e_n,one})$ \label{item-7-unit}
    outputs an allocation where bidder $1$ wins a bundle  not containing $e_2$, then he pays zero.
    \end{enumerate} 
\end{lemma}
Lemma \ref{lemma-small-pay-unit} is a direct consequence of the properties of the mechanism. 
We defer the proof to Subsection \ref{subsec-lemma-proof}. We can now prove Lemma \ref{lemma-case-1-unit} and Lemma \ref{lemma-case-2-unit}.

\subsection{Proof of Lemma \ref{lemma-case-1-unit}}  \label{subsubsec-case1-unit}
The proof is as follows. We first show that the strategy $\mathcal S_1$ always dictates the same message for $v_1^{both}$ and   
$v_1^{e_1,big}$ (Claim \ref{claim-sameM-both-large1-unit}). We proceed by demonstrating that in the case under consideration in Lemma \ref{lemma-case-1-unit}, the strategy
$\mathcal S_1$ dictates the same message for the valuations $v_1^{both}$ and $v_1^{e_2,big}$ (Claim \ref{claim-sameM-both-large2-unit}).
We remind that $\mathcal V_1=\{v_1^{e_1,one},v_1^{e_1,big},v_1^{e_2,big},v_1^{both}\}$, so combining
Claims \ref{claim-e1-same-unit}, \ref{claim-sameM-both-large1-unit} and \ref{claim-sameM-both-large2-unit} gives that the strategy $\mathcal S_1$ dictates the same message for all the valuations in $\mathcal V_1$, as needed. 

\begin{claim}\label{claim-sameM-both-large1-unit}
The strategy $\mathcal S_1$ dictates the same message at vertex $u$ for both $v_1^{both}$  and $v_1^{e_1,big}$.
\end{claim}
\begin{proof}
Note that by Lemma \ref{lemma-small-pay-unit} part \ref{item-4-unit}, given $(v_1^{both},v_2^{e_2,one},v_3^{e_3,one},\ldots\allowbreak,v_n^{e_n,one})$, bidder $1$ wins a bundle that contains $e_1$ and pays at most $1$. Therefore:
\begin{equation}\label{eq-e1toboth}
    v_1^{e_1,big}(f(v_1^{both},v_2^{e_2,one},v_3^{e_3,one},\ldots,v_n^{e_n,one}))-P_1(v_1^{both},v_2^{e_2,one},v_3^{e_3,one},\ldots,v_n^{e_n,one})\ge 3k^4-1
\end{equation}
Whereas by Lemma \ref{lemma-small-pay-unit} part \ref{item-2-unit}:
\begin{equation}\label{eq-e1toboth2}
    v_1^{e_1,big}(f(v_1^{e_1,big},v_2^{e_1,big},v_3^{e_3,one},\ldots,v_n^{e_n,one}))-P_1(v_1^{e_1,big},v_2^{e_1,big},v_3^{e_3,one},\ldots,v_n^{e_n,one})\le 3k^4-2k^3-k^2
\end{equation}
Combining (\ref{eq-e1toboth}) and (\ref{eq-e1toboth2}) gives:
\begin{multline*}
    v_1^{e_1,big}(f(v_1^{e_1,big},v_2^{e_1,big},v_3^{e_3,one},\ldots,v_n^{e_n,one}))-P_1(v_1^{e_1,big},v_2^{e_1,big},v_3^{e_3,one},\ldots,v_n^{e_n,one})< \\ v_1^{e_1,big}(f(v_1^{both},v_2^{e_1,big},v_3^{e_3,one},\ldots,v_n^{e_n,one}))-P_1(v_1^{both},v_2^{e_1,big},v_3^{e_3,one},\ldots,v_n^{e_n,one})
\end{multline*}
We remind that vertex $u$ belongs in $Path(\mathcal S_1(v_1^{e_1,big}),\mathcal S_2(v_2^{e_1,big}),S_3(v_3^{e_3,one}),\ldots,\mathcal S_n(v_n^{e_n,one}))$ and also in
$Path(\mathcal{S}_1(v_1^{both}),\mathcal{S}_2(v_2^{e_2,one}), \mathcal{S}_3(v_3^{e_3,one}),\ldots,\mathcal{S}_n(v_n^{e_n,one}))$, so Lemma \ref{lemma-bad-leaf-good-leaf} gives that the strategy $\mathcal S_1$ dictates the same message for  $v_1^{e_1,big}$ and $v_1^{both}$ at vertex $u$.
\end{proof}
\begin{claim}\label{claim-sameM-both-large2-unit}
If $f(v_1^{both},v_2^{e_1,big},v_3^{e_3,one},\ldots,v_n^{e_n,one})$ outputs an allocation where bidder $1$ wins a bundle that contains $e_2$, 
 then   the strategy $\mathcal S_1$ dictates the same message at vertex $u$ for both $v_1^{both}$  and $v_1^{e_2,big}$.
\end{claim}
\begin{proof}
    By Lemma \ref{lemma-small-pay-unit} part \ref{item-6-unit} and by assumption, given $(v_1^{both},v_2^{e_1,big},v_3^{e_3,one},\ldots,v_n^{e_n,one})$, bidder $1$  wins $e_2$ and pays at most $2k^2$, so:
    \begin{equation}\label{eq-case1-claim2-1-unit-first}
        v_1^{e_2,big}(f(v_1^{both},v_2^{e_1,big},v_3^{e_3,one},\ldots,v_n^{e_n,one}))-P_1(v_1^{both},v_2^{e_1,big},v_3^{e_3,one},\ldots,v_n^{e_n,one})\ge 3k^4-2k^2
    \end{equation}
    Whereas by Lemma \ref{lemma-small-pay-unit} part \ref{item-3-unit}:
     \begin{equation}\label{eq-case1-claim2-2-unit-first}
        v_1^{e_2,big}(f(v_1^{e_2,big},v_2^{e_2,big},v_3^{e_3,one},\ldots,v_n^{e_n,one}))-P_1(v_1^{e_2,big},v_2^{e_2,big},v_3^{e_3,one},\ldots,v_n^{e_n,one})\le 3k^4-2k^3-k^2
    \end{equation}
Combining (\ref{eq-case1-claim2-1-unit-first}) and (\ref{eq-case1-claim2-2-unit-first}) gives:
\begin{multline*}
    v_1^{e_2,big}(f(v_1^{e_2,big},v_2^{e_2,big},v_3^{e_3,one},\ldots,v_n^{e_n,one}))-P_1(v_1^{e_2,big},v_2^{e_2,big},v_3^{e_3,one},\ldots,v_n^{e_n,one})< \\
    v_1^{e_2,big}(f(v_1^{both},v_2^{e_1,big},v_3^{e_3,one},\ldots,v_n^{e_n,one}))-P_1(v_1^{both},v_2^{e_1,big},v_3^{e_3,one},\ldots,v_n^{e_n,one})
\end{multline*}
Similarly to before, vertex $u$ is in $Path(\mathcal S_1(v_1^{e_2,big}),\mathcal S_2(v_2^{e_2,big}),\mathcal{S}_3(v_3^{e_3,one}),\ldots,\mathcal{S}_n(v_n^{e_n,one}))$ and also in $Path(\mathcal S_1(v_1^{both}),\allowbreak \mathcal S_2(v_2^{e_1,big}),\mathcal{S}_3(v_3^{e_3,one}),\ldots,\mathcal{S}_n(v_n^{e_n,one}))$, so applying Lemma \ref{lemma-bad-leaf-good-leaf} finishes the proof. 
\end{proof}

\subsection{Proof of Lemma \ref{lemma-case-2-unit}} \label{subsubsec-case2-unit}
Similarly to the previous case, the proof for this case also requires two steps. Both steps rely on the assumption that bidder $1$ does not win $e_2$ given the valuation profile $(v_1^{both},v_2^{e_1,big},v_3^{e_3,one},\ldots,v_n^{e_n,one})$.
 We first show that the strategy $\mathcal S_1$ dictates the same message for $v_1^{both}$ and 
$v_1^{e_1,one}$ (Claim \ref{claim-sameM-both-small-e1-unit}). We proceed by demonstrating that $\mathcal S_1$ dictates the same message for both $v_1^{both}$ and $v_1^{e_2,big}$ (Claim \ref{claim-sameM-both-large-e2-unit}). Combining  Claims \ref{claim-sameM-both-small-e1-unit} and \ref{claim-sameM-both-large-e2-unit} with Claim \ref{claim-e1-same-unit} gives that $\mathcal{S}_1$ sends the same message for all the valuations in $\mathcal V_1$, that completes the proof.

\begin{claim}\label{claim-sameM-both-small-e1-unit}
          If $f(v_1^{both},v_2^{e_1,big},v_3^{e_3,one},\ldots,v_n^{e_n,one})$
    outputs an allocation where bidder $1$ wins a bundle not containing $e_2$,   
then the strategy $\mathcal S_1$ dictates the same message at vertex $u$ for both $v_1^{both}$  and $v_1^{e_1,one}$.
\end{claim}
\begin{proof}
Note that by Lemma \ref{lemma-small-pay-unit} part \ref{item-1-unit}, given $(v_1^{e_1,one},v_2^{e_2,one},\ldots\allowbreak,v_n^{e_n,one})$, bidder $1$ wins a bundle that contains $e_1$ and pays at most $1$. Therefore:
\begin{equation}\label{eq-case2-e1to1-unit}
    v_1^{both}(f(v_1^{e_1,one},v_2^{e_2,one},\ldots\allowbreak,v_n^{e_n,one}))-P_1(v_1^{e_1,one},v_2^{e_2,one},\ldots\allowbreak,v_n^{e_n,one})\ge 2k^2+1
\end{equation}
Whereas by assumption and by Lemma \ref{lemma-small-pay-unit} parts \ref{item-new-unit} and \ref{item-7-unit}: 
\begin{equation}\label{eq-bothnothing-unit}
    v_1^{both}(f(v_1^{both},v_2^{e_1,big},v_3^{e_3,one},\ldots,v_n^{e_n,one}))-P_1(v_1^{both},v_2^{e_1,big},v_3^{e_3,one},\ldots,v_n^{e_n,one})=0
\end{equation}
Combining (\ref{eq-case2-e1to1-unit}) and (\ref{eq-bothnothing-unit}) gives:
\begin{multline*}
    v_1^{both}(f(v_1^{both},v_2^{e_1,big},v_3^{e_3,one},\ldots,v_n^{e_n,one}))-P_1(v_1^{both},v_2^{e_1,big},v_3^{e_3,one},\ldots,v_n^{e_n,one})< \\ v_1^{both}(f(v_1^{e_1,one},v_2^{e_2,one},\ldots\allowbreak,v_n^{e_n,one}))-P_1(v_1^{e_1,one},v_2^{e_2,one},\ldots\allowbreak,v_n^{e_n,one})
\end{multline*}
We remind that by assumption vertex $u$ belongs in $Path(\mathcal S_1(v_1^{e_1,one}),\mathcal S_2(v_2^{e_2,one}),,\ldots,\mathcal S_n(v_n^{e_n,one}))$ and also in
$Path(\mathcal{S}_1(v_1^{both}),\mathcal{S}_2(v_2^{e_1,big}), \mathcal{S}_3(v_3^{e_3,one}),\ldots,\mathcal{S}_n(v_n^{e_n,one}))$, so applying Lemma \ref{lemma-bad-leaf-good-leaf} completes the proof.
\end{proof}

\begin{claim}\label{claim-sameM-both-large-e2-unit}
          If $f(v_1^{both},v_2^{e_1,big},v_3^{e_3,one},\ldots,v_n^{e_n,one})$
    outputs an allocation where bidder $1$  wins a bundle not containing $e_2$, 
then      the strategy $\mathcal S_1$ dictates the same message at vertex $u$ for both $v_1^{both}$  and $v_1^{e_2,big}$.
\end{claim}
\begin{proof}
    By Lemma \ref{lemma-small-pay-unit} parts \ref{item-new-unit} and  \ref{item-7-unit} and by assumption, given $(v_1^{both},v_2^{e_1,big},v_3^{e_3,one},\ldots,v_n^{e_n,one})$, bidder $1$ gets  nothing and pays nothing, so:
    \begin{equation}\label{eq-case1-claim2-1-unit}
        v_1^{both}(f(v_1^{both},v_2^{e_1,big},v_3^{e_3,one},\ldots,v_n^{e_n,one}))-P_1(v_1^{both},v_2^{e_1,big},v_3^{e_3,one},\ldots,v_n^{e_n,one})=0
    \end{equation}
    Whereas by Lemma \ref{lemma-small-pay-unit} part \ref{item-5-unit}:
     \begin{equation}\label{eq-case1-claim2-2-unit}
        v_1^{both}(f(v_1^{e_2,big},v_2^{e_2,one},\ldots,v_n^{e_n,one}))-P_1(v_1^{e_2,big},v_2^{e_2,one},\ldots,v_n^{e_n,one})\ge 2k^2-k^2=k^2
    \end{equation}
Combining (\ref{eq-case1-claim2-1-unit}) and (\ref{eq-case1-claim2-2-unit}) gives:
\begin{multline*}
v_1^{both}(f(v_1^{both},v_2^{e_1,big},v_3^{e_3,one},\ldots,v_n^{e_n,one}))-P_1(v_1^{both},v_2^{e_1,big},v_3^{e_3,one},\ldots,v_n^{e_n,one})< \\
    v_1^{both}(f(v_1^{e_2,big},v_2^{e_2,one},\ldots,v_n^{e_n,one}))-P_1(v_1^{e_2,big},v_2^{e_2,one},\ldots,v_n^{e_n,one})
\end{multline*}
Similarly to before, vertex $u$ is in $Path(\mathcal S_1(v_1^{both}),\mathcal S_2(v_2^{e_1,big}),\mathcal{S}_3(v_3^{e_3,one}),\ldots,\mathcal{S}_n(v_n^{e_n,one}))$ and also in $Path(\mathcal S_1(v_1^{e_2,big}),\allowbreak \mathcal S_2(v_2^{e_2,one}),\mathcal{S}_3(v_3^{e_3,one}),\ldots,\mathcal{S}_n(v_n^{e_n,one}))$, so applying Lemma \ref{lemma-bad-leaf-good-leaf} completes the proof. 
\end{proof}

\section{Missing Proofs} \label{appsec-missing-proofs}
\subsection{Proofs of Basic Observations about Mechanisms 
(Lemma \ref{lemma-small-pay}, \ref{lemma-small-pay-add} and \ref{lemma-small-pay-unit})} \label{subsec-small-pay-proof-alltogether}
The deduction of all components in the statements of all three lemmas is a direct consequence of the fact that the mechanisms in all of them are  obviously strategy-proof (and thus also dominant-strategy incentive compatible), satisfy individual rationality and no-negative-transfers and provide welfare approximation better than
$\min\{m,n\}$.
We write them  for the sake of completeness.

In all the following lemmas, we say for abbreviation that certain inequalities hold  because of individual rationality instead of saying that they hold because the allocation rule together with the payment scheme are realized by a mechanism and strategies that satisfy individual rationality. We do the same for the no-negative-transfers property.


\subsubsection{Proof of Lemma \ref{lemma-small-pay}} \label{subsec-small-pay-proof}
We slightly abuse notation throughout the proof by writing 
$f(v_1,v_2)$ for $f(v_1,v_2,v_3^{one},\ldots,v_n^{one})$ and $P_1(v_1,v_2)$ for $P_1(v_1,v_2,v_3^{one},\ldots,v_n^{one})$ for every pair of valuations $v_1\in V_1$ and $v_2\in V_2$.  

For part \ref{item-1}, note that because of individual rationality:
\begin{equation}\label{eq-1}
 v_1^{one}(f(v_1^{one},v_2^{one}))-P_1(v_1^{one},v_2^{one})\ge 0   
\end{equation}
We remind that by assumption $f$ allocates to bidder $1$ at least one item given $(v_1^{one},v_2^{one})$, so:
\begin{equation}\label{eq-2}
 v_1^{one}(f(v_1^{one},v_2^{one}))=1   
\end{equation}
Combining (\ref{eq-1}) and (\ref{eq-2}) gives part \ref{item-1}.

For part \ref{item-2}, note that given $(v_1^{ONE},v_2^{all})$, player $1$ gets the empty bundle because $f$ gives approximation better than $\min\{m,n\}$ only if player $2$ wins all the items. Because of individual rationality: 
$$v_1^{ONE}(f(v_1^{ONE},v_2^{all}))-P_1(v_1^{ONE},v_2^{all})\ge 0 $$
Since $v_1^{ONE}(f(v_1^{ONE},v_2^{all}))=0$, we get that $0\ge P_1(v_1^{ONE},v_2^{all})$, and because of the no-negative-transfers property, $P_1(v_1^{ONE},v_2^{all})=0$, as needed.  

To prove part \ref{item-3}, we define another valuation:
$$
\hat{v_1}(s)=\begin{cases}
    k^2 \quad s=m, \\
    0 \quad \text{otherwise.}
\end{cases}
$$
Given $(\hat{v_1},v_2^{one})$, player $1$ wins $m$ items because of the guaranteed approximation ratio of $f$. The inequality $\hat{v_1}(f(\hat{v_1},v_2^{one}))-P_1(\hat{v_1},v_2^{one})\ge 0$ holds because of individual rationality, and therefore $P_1(\hat{v_1},v_2^{one})\le k^2$. 

Note that $f$ clearly also allocates all items to player $1$ given $({v_1}^{all},v_2^{one})$ because of its approximation guarantee. The fact that the mechanisms is dominant-strategy incentive compatible and $f$ outputs the same allocation for both $(\hat{v_1},v_2^{one})$ and  $({v_1}^{all},v_2^{one})$ implies that $P_1({v_1}^{all},v_2^{one})=P_1(\hat{v_1},v_2^{one})$, so $P_1({v_1}^{all},v_2^{one})$ is also smaller than $k^2$, which completes the proof.

\subsubsection{Proof of Lemma \ref{lemma-small-pay-add}} \label{appsubsec-proof-lemma-add-small-pay}
We slightly abuse notation throughout the proof by writing 
$f(v_1,v_2)$ for $f(v_1,v_2,v_3^{e_3,one},\ldots,v_n^{e_n,one})$ and $P_1(v_1,v_2)$ for $P_1(v_1,v_2,v_3^{e_3,one},\ldots,v_n^{e_n,one})$ for every pair of valuations $v_1\in V_1$ and $v_2\in V_2$.  
 
 For part \ref{item-1-add}, note that because of individual rationality:
\begin{equation}\label{eq-1-add}
v_1^{e_1,one}(f(v_1^{e_1,one},v_2^{e_2,one}))-P_1(v_1^{e_1,one},v_2^{e_2,one})\ge 0
\end{equation}
We remind that $f$ allocates a bundle that contains item $e_1$ to bidder $1$  given $(v_1^{e_1,one},v_2^{e_2,one})$, so:
\begin{equation}\label{eq-2-add}
 v_1^{e_1,one}(f(v_1^{e_1,one},\allowbreak v_2^{e_2,one}))=1   
\end{equation}
Combining (\ref{eq-1-add}) and (\ref{eq-2-add}) gives part \ref{item-1-add}. 

For part \ref{item-2-add}, note that the approximation guarantee of $f$ implies that given $(v_1^{e_1,big},v_2^{e_1,big})$, either bidder $1$ or bidder $2$ wins $e_1$. If bidder $2$ wins $e_1$, then $v_1^{e_1,big}(f(v_1^{e_1,big},v_2^{e_1,big})=0$, so this inequality combined with the fact that the mechanism is individually rational gives that: 
\begin{align*}
    v_1^{e_1,big}(f(v_1^{e_1,big},v_2^{e_1,big}))-P_1(v_1^{e_1,big},v_2^{e_1,big})&\ge 0 \\ \implies 
    P_1(v_1^{e_1,big},v_2^{e_1,big}) &\le 0
\end{align*}
Because of no-negative-transfers, $P_1(v_1^{e_1,big},v_2^{e_1,big})=0$.

We now explain the other case in which  $f(v_1^{e_1,big},v_2^{e_1,big})$ outputs an allocation where bidder $1$ wins $e_1$. To this end, we define the following valuation of bidder $1$:
$$
\hat{v}_1(x)=\begin{cases}
        {2k^3}+k^2 \quad x=e_1, \\
    0 \quad \text{otherwise.}
\end{cases}
$$
Observe that because of the approximation guarantee, $f(\hat{v}_1,v_2^{e_1,big})$  outputs an allocation such that bidder $2$ gets $e_1$, so $\hat{v}_1(f(\hat{v}_1,v_2^{e_1,big}))=0$.  
The fact that the mechanism is individually rational and satisfies no-negative-transfers implies that $P_1(\hat{v}_1,v_2^{e_1,big})=0$. We remind that the allocation rule $f$ is realized by a dominant-strategy mechanism, so we also have that:
\begin{align*}
    \hat{v}_1(f(\hat{v}_1,v_2^{e_1,big}))-P_1(\hat{v}_1,v_2^{e_1,big}))&\ge 
\hat{v}_1(f(v_1^{e_1,big},v_2^{e_1,big})-P_1(v_1^{e_1,big},v_2^{e_1,big}) \\
\implies 0 &\ge 2k^3+k^2- P_1(v_1^{e_1,big},v_2^{e_1,big}) &\text{(by assumption)}
\end{align*}
so $P_1(v_1^{e_1,big},v_2^{e_1,big})\ge {2k^3}+k^2$, which completes the proof of part \ref{item-2-add}. The proof of part \ref{item-3-add} is analogous. 

We now prove part \ref{item-4-add}. We will first show that $f(v_1^{both},v_2^{e_2one})$ allocates bidder $1$ a bundle that contains item $e_1$. 
To this end, observe that:
\begin{equation}\label{eq-part4-add}
\begin{aligned}
    v_1^{both}(f(v_1^{both},v_2^{e_2,one}))-P_1(v_1^{both},v_2^{e_2,one})&\ge v_1^{both}(f(v_1^{e_1,one},v_2^{e_2,one}))-P_1(v_1^{e_1,one},v_2^{e_2,one})  \\
    &\ge 2k^2+1 
\end{aligned}    
\end{equation}
The first inequality holds because the mechanism $\mathcal M$ is a dominant-strategy mechanism, and the second one holds because of part \ref{item-1-add} that we previously proved.  
Combining (\ref{eq-part4-add}) with the property of no-negative-transfers implies that $v_1^{both}(f(v_1^{both},v_2^{e_2,one}))\ge 2k^2+1$.
Note that $v_1^{both}$ has a value of at most $2k^2$ to bundles that do not contain $e_1$, 
so it has to be the case that 
$f(v_1^{both},v_2^{e_2,one})$ 
allocates item $e_1$ to bidder $1$.
For the upper bound on the payment, observe that because of individual rationality:
$$
4k^2+2-P_1(v_1^{both},v_2^{e_2,one})\ge v_1^{both}(f(v_1^{both},v_2^{e_2,one}))-P_1(v_1^{both},v_2^{e_2,one})\ge 0  
$$
so indeed $P_1(v_1^{both},v_2^{e_2,one})\le 4k^2+2$.

We turn our attention to part \ref{item-5-add}. To prove it,
we define another valuation:
$$
\tilde{v_1}(x)=\begin{cases}
    k^2 \quad x=e_2, \\
    0 \quad \text{otherwise.}
\end{cases}
$$
Given $(\tilde{v_1},v_2^{e_2,one})$, bidder $1$ wins a bundle that contains $e_2$
 because of the guaranteed approximation ratio of $f$. Due to the individual rationality property: 
 \begin{equation}\label{eq-part5-small-p}
  \tilde{v_1}(f(\tilde{v_1},v_2^{e_2,one}))-P_1(\tilde{v_1},v_2^{e_2,one})\ge 0 \implies k^2 \ge P_1(\tilde{v_1},v_2^{e_2,one})   
 \end{equation}
Note that $f$ clearly also allocates a bundle that contains item $e_2$ to bidder $1$ given $({v_1}^{e_2,big},v_2^{e_2,one})$ because of its approximation guarantee. By dominant-strategy incentive compatibility, we have that:
\begin{align*}
    {v_1}^{e_2,big}(f({v_1}^{e_2,big},v_2^{e_2,one})) - P_1(v_1^{e_2,big},v_2^{e_2,one}) \ge
    {v_1}^{e_2,big}(f(\tilde v_1,v_2^{e_2,one})) - P_1(\tilde v_1,v_2^{e_2,one})
\end{align*}
The allocation rule $f$
outputs an allocation where bidder $1$ wins item $e_2$ for both  $(\tilde{v_1},v_2^{e_2,one})$ and  $({v_1}^{e_2,big},v_2^{e_2,one})$, so  ${v_1}^{e_2,big}(f({v_1}^{e_2,big},v_2^{e_2,one}))={v_1}^{e_2,big}(f(\tilde v_1,v_2^{e_2,one}))$. By that,  $P_1(\tilde v_1,v_2^{e_2,one})\ge P_1(v_1^{e_2,big},v_2^{e_2,one})$. Combining this inequality with \ref{eq-part5-small-p} gives that $P_1(v_1^{e_2,big},v_2^{e_2,one})\le k^2$, as needed.

To prove parts \ref{item-new-add},\ref{item-6-add} and \ref{item-7-add}, we analyze the valuation profile $(v_1^{both},v_2^{e_1,big})$. Note that because of the approximation guarantee of $f$, it allocates item $e_1$ to bidder $2$, so we have part \ref{item-new-add}. For part \ref{item-6-add}, assume that  $f(v_1^{both},v_2^{e_1,big})$ allocates $e_2$ to bidder $1$. Since only items $e_1$ and $e_2$ are valuable for the valuation $v_1^{both}$, we have that $v_1^{both}(f(v_1^{both},v_2^{e_1,big}))=2k^2$, so by individual rationality, $P_1(v_1^{both},v_2^{e_1,big})\le 2k^2$, as needed. 

Analogously, for part \ref{item-7-add}: if bidder $1$ does not win $e_2$, then he wins no items that are valuable for him, so in this case $v_1^{both}(f(v_1^{both},v_2^{e_1,big}))=0$. By individual rationality and no-negative-transfers, we get that  $P_1(v_1^{both},v_2^{e_1,big})=0$, as needed.




\subsubsection{Proof of Lemma \ref{lemma-small-pay-unit}} \label{subsec-lemma-proof}
Similarly to the proof of Lemma \ref{lemma-small-pay-unit}, we  abuse notation  by writing 
$f(v_1,v_2)$ for $f(v_1,v_2,v_3^{e_3,one},\ldots,\allowbreak v_n^{e_n,one})$ and $P_1(v_1,v_2)$ for $P_1(v_1,v_2,v_3^{e_3,one},\ldots,v_n^{e_n,one})$ for every pair of valuations $v_1\in V_1$ and $v_2\in V_2$. 

The proof of parts \ref{item-1-unit}, \ref{item-2-unit}, \ref{item-3-unit} and \ref{item-5-unit} are identical to the proofs of these parts in Lemma \ref{lemma-small-pay-add}, since all the valuations in all of these cases are both unit demand and additive. 
We now prove the remaining parts. 

For part \ref{item-4-unit}, we first show that bidder $1$ wins a bundle that contains the item $e_1$. Note that since the allocation rule $f$ and the payment scheme $P_1$ are realized by a dominant-strategy mechanism, we have that:
\begin{equation}\label{eq-part4-unit}
\begin{aligned}
    v_1^{both}(f(v_1^{both},v_2^{e_2,one}))-P_1(v_1^{both},v_2^{e_2,one})&\ge v_1^{both}(f(v_1^{e_1,one},v_2^{e_2,one}))-P_1(v_1^{e_1,one},v_2^{e_2,one}) \\
    &\ge 2k^2+1 &\text{(by part \ref{item-1-unit})}
\end{aligned}    
\end{equation}
Combining (\ref{eq-part4-unit}) with the property of no-negative-transfers implies that $v_1^{both}(f(v_1^{both},v_2^{e_2,one}))\ge 2k^2+1$.
Note that $v_1^{both}$ has a value of at most $2k^2$ to bundles that do not contain $e_1$, so $f(v_1^{both},v_2^{e_2,one})$ necessarily outputs an allocation where bidder $1$ wins a bundle that contains $e_1$. 

For the upper bound on the payment, note that $v_1^{both}(f(v_1^{both},v_2^{e_2,one}))=v_1^{both}(f(v_1^{e_1,one},v_2^{e_2,one}))$ and that by part \ref{item-1-unit}, $P_1(v_1^{e_1,one},v_2^{e_2,one})\le 1$. Combining these inequalities with (\ref{eq-part4-unit}) gives that $P_1(v_1^{both},v_2^{e_2,one})\le 1$, that completes the proof of part \ref{item-4-unit}. 

To prove parts \ref{item-new-unit},\ref{item-6-unit} and \ref{item-7-unit}, we analyze the valuation profile $(v_1^{both},v_2^{e_1,big})$. Note that because of the approximation guarantee of $f$, it allocates item $e_1$ to bidder $2$, so we have part \ref{item-new-add}. For part \ref{item-6-add}, assume that  $f(v_1^{both},v_2^{e_1,big})$ allocates $e_2$ to bidder $1$. Since only items $e_1$ and $e_2$ are valuable for the valuation $v_1^{both}$, we have that $v_1^{both}(f(v_1^{both},v_2^{e_1,big}))=2k^2$, so by individual rationality, $P_1(v_1^{both},v_2^{e_1,big})\le 2k^2$, as needed. 

Analogously, for part \ref{item-7-add}: if bidder $1$ does not win $e_2$, then he wins no items that are valuable for him, so in this case $v_1^{both}(f(v_1^{both},v_2^{e_1,big}))=0$. By individual rationality and no-negative-transfers, we get that  $P_1(v_1^{both},v_2^{e_1,big})=0$, as needed.
\subsection{Proof  of Proposition \ref{prop-obs}}\label{subsec-prop-proof}
Showing that an obviously strategy-proof implementation is a dominant strategy implementation is trivial, so we only show the other direction: that a dominant-strategy implementation is, in fact, an obviously strategy-proof implementation. 

Fix a mechanism together with the dominant strategies $\mathcal S_1,\ldots,\mathcal S_n$ that realize an allocation rule  $f$ together with payment schemes $P_1,\ldots,P_n$.   

Assume towards a contradiction that there exists a player $i$ such that the
strategy $\mathcal S_i$ is not obviously dominant. Thus, by definition, there exists a valuation $v_i$ such that the behavior $\mathcal S_i(v_i)$ is not obviously dominant for it. It implies that there is a vertex $u\in \mathcal N_i$ that is attainable given the behavior $\mathcal S_i(v_i)$ and behavior profiles 
$(B_1',\ldots,B_n') \in \mathcal B_1\times\cdots\times \mathcal B_n$ and
$ B_{-i}\in \mathcal B_{-i}$ 
such that:
\begin{equation}\label{eq-obs-dominant}
v_i(f_i(\mathcal S_i(v_i),B_{-i}))-p_i(\mathcal S_i(v_i),B_{-i})<v_i(f_i(B_i',B_{-i}'))-p_i(B_i',B_{-i}')    
\end{equation}
where $u\in Path(\mathcal S_i(v_i),B_{-i})\cap Path(B_{i}',B_{-i}')$ and $\mathcal S_i(v_i)$ and $B_i'$ dictate different messages at vertex $u$. 

To reach a contradiction, we construct a strategy profile for the bidders in $N\setminus \{i\}$, which we denote with $\mathcal S_{-i}'$. This strategy profile is \textquote{constant}, in the sense that for every valuation profile $v_{-i}\in V_{-i}$, it outputs the same behavior profile, which we denote with $B_{-i}''$.  

The behavior profile  $B_{-i}''$ is as follows. First, for every vertex that is not a descendant of $u$, $B_{-i}''$ outputs the same messages as the messages specified by $B_{-i}$. We now describe the behavior of $B_{-i}''$ in the subtree  rooted at vertex $u$. For that,
let us denote with $n,n'$ be the subsequent nodes of $u$ given the message that the behaviors $\mathcal S_i(v_i)$ and $B_i'$ dictate, respectively. For all the vertices at the subtree rooted at vertex $n$,  the behavior profile
$B_{-i}''$ specifies the same message as $B_{-i}$. 
Similarly, at the subtree rooted at vertex $n'$, the behavior profile $B_{-i}''$ outputs at every vertex the same message as $B_{-i}'$ dictates. For every other vertex, $B_{-i}''$ outputs some arbitrary message.

Fix an arbitrary valuation profile $v_{-i}$ of the players in $N\setminus\{i\}$.  
The construction of the strategy profile $\mathcal S_{-i}'$ guarantees that the behavior profile $(\mathcal S_i(v_i),\mathcal S_{-i}'(v_{-i}))$ reaches the same leaf as $(\mathcal S_i(v_i),B_{-i})$, so $f_i(\mathcal S_i(v_i),B_{-i}))=f_i(\mathcal S_i(v_i),\mathcal S_{-i}'(v_{-i}))$ and $p_i(\mathcal S_i(v_i),B_{-i})=p_i(\mathcal S_i(v_i), \mathcal S_{-i}'(v_{-i}))$. 
Denote with $\mathcal S_i'$ be the strategy of player $i$ that outputs for every $v_i\in V_i$ the behavior $B_i'$. Similarly to before, $f_i(B_i',B_{-i}')=f_i(\mathcal S_i'(v_i),\mathcal S_{-i}'(v_{-i}))$ and $p_i(B_i',B_{-i}')=p_i(\mathcal S_i'(v_i),\mathcal S_{-i}'(v_{-i})$. Plugging these equalities in (\ref{eq-obs-dominant}) gives:
$$
v_i(f_i(\mathcal{S}_i(v_i),\mathcal{S}_{-i}'(v_{-i}))) - p_i(\mathcal{S}_i(v_i),\mathcal{S}_{-i}'(v_{-i})) <v_i(f_i(\mathcal{S}_i'(v_i),\mathcal{S}_{-i}'(v_{-i}))) - p_i(\mathcal{S}_i'(v_i),\mathcal{S}_{-i}'(v_{-i}))
$$
So by definition, the strategy $\mathcal S_i$ is not  dominant, so we get a contradiction, completing the proof. 
\end{document}